\documentclass[a4paper,twocolumn,accepted=2021-11-22,amsfonts,amssymb,amsmath]{quantumarticle}
\pdfoutput=1

\usepackage[utf8]{inputenc}
\usepackage[english]{babel}
\usepackage[T1]{fontenc}

\usepackage{cancel}

\usepackage{bbold}

\newcommand{\id}[0]{\mathbb{1}}
\newcommand{\trans}[0]{^{\rm T}}
\newcommand{\reals}{\mathbb{R}}

\newcommand{\Csa}[1]{\mathbf{H}_{#1}(\mathbb{C})}
\newcommand{\Csap}[1]{\mathbf{H}_{#1}^{+}(\mathbb{C})}

\DeclareMathOperator{\conv}{conv}
\DeclareMathOperator{\diag}{diag}
\DeclareMathOperator{\tr}{tr}

\usepackage{amsthm}
\newtheorem{theorem}{Theorem}
\newtheorem{lemma}[theorem]{Lemma}

\newtheorem{definition}{Definition}
\newtheorem{example}{Example}

\usepackage{enumerate}

\usepackage{soul}
\usepackage{braket}
\usepackage{framed}
\usepackage{graphicx}
\newcommand{\caphead}[1]{\textbf{#1}}

\usepackage{hyperref}
\usepackage[capitalise]{cleveref}
\usepackage[numbers,sort&compress]{natbib}

\usepackage{csquotes}

\usepackage[symbol]{footmisc}

\usepackage{xcolor}
\definecolor{shadecolor}{rgb}{0.9,0.9,0.9}

\newcommand{\IQOQI}{Institute for Quantum Optics and Quantum Information, Austrian Academy of Sciences, Boltzmanngasse 3, A-1090 Vienna, Austria}
\newcommand{\Peri}{Perimeter Institute for Theoretical Physics, 31 Caroline Street North, Waterloo, ON N2L 2Y5, Canada}

\begin{document}
\title{Quantum Darwinism and the spreading of classical information in non-classical theories}
\author{Roberto D.\ Baldij\~ao$^{*,}$}
\affiliation{Instituto de F\'isica Gleb Wataghin, Universidade Estadual de Campinas, Campinas, SP 13083-859, Brazil}
\affiliation{\IQOQI{}}
\orcid{0000-0002-4248-2015}
\author{Marius Krumm$^{*,}$}
\affiliation{\IQOQI{}}
\affiliation{Faculty of Physics, University of Vienna, Boltzmanngasse 5, A-1090 Vienna, Austria}
\author{Andrew J.\ P.\ Garner}
\affiliation{\IQOQI{}}
\orcid{0000-0002-3747-9997}
\author{Markus P.\ M\"uller}
\affiliation{\IQOQI{}}
\affiliation{Vienna Center for Quantum Science and Technology (VCQ), Faculty of Physics, University of Vienna, Vienna, Austria}
\affiliation{\Peri{}}
\orcid{0000-0002-8086-5586}

\begin{abstract}
Quantum Darwinism posits that the emergence of a classical reality relies on the spreading of classical information from a quantum system to many parts of its environment. 
But what are the essential physical principles of quantum theory that make this mechanism possible?
We address this question by formulating the simplest instance of Darwinism -- CNOT-like {\em fan-out} interactions -- in a class of probabilistic theories that contain classical and quantum theory as special cases. 
We determine necessary and sufficient conditions for any theory to admit such interactions. 
We find that every theory with non-classical features that admits this idealized spreading of classical information must have both entangled states and entangled measurements. 
Furthermore, we show that Spekkens' toy theory admits this form of Darwinism, and so do all probabilistic theories that satisfy principles like strong symmetry, or contain a certain type of decoherence processes. 
Our result suggests the counter-intuitive general principle that in the presence of local non-classicality, a classical world can only emerge if this non-classicality can be \mbox{``amplified''} to a form of entanglement.
\end{abstract}
\maketitle

\section{Introduction}
\label{sec: Intro}

\renewcommand*{\thefootnote}{\fnsymbol{footnote}}
\footnotetext{\textbf{These authors contributed equally to this work.}}
\renewcommand*{\thefootnote}{\arabic{footnote}}
\setcounter{footnote}{0}

{\em Quantum Darwinism}~\cite{Zurek_QuantumOrigins,Zurek_QDReview1,Zurek_QDReview2,Brandao_QD,Knott_QDInfiniteDimension,Horodecki_NoMI,Le_ObjectivityQDxSSB,Blume-Kohout_QDBrownian,Blume-Kohout_RedundantInfo,Riedel_QDIllumination,Zwolak_NoisyChannels,Zwolak_MixedEnvironment} addresses one of the toughest questions raised by quantum theory:
 If the universe is fundamentally described by quantum mechanics, how does an objective classical world arise?
At the heart of this question is a tension between the microscopic quantum realm, in which systems happily exist in states of super-imposed possibility,
   and the macroscopic world of ``classical'' systems (such as the pointer needle of a read-out gauge), which are only ever observed in definite objective states.
Several mechanisms and formalisms have been proposed which intend to provide a bridge between the quantum and classical realms, including the formal limit of $\hbar\to 0$~\cite{DiracBook}, saddle point approximations to the path integral~\cite{Sakurai}, and the process of environment-induced decoherence~\cite{Zurek_QuantumOrigins,Schlosshauer_Deco}.

Quantum Darwinism identifies {a} key prerequisite for such a bridge to arise:
 there must be a mechanism by which some aspect of a quantum system can be spread out to many parts of its environment.
Particularly, since the no--cloning theorem~\cite{Wootters_NoCloning} forbids the copying of quantum information,
 this means some classical information from the system must be copied into its environment in such a way that
 given long enough (and enough of the environment), this information can be learned through enough measurements on the environment.

Here we ask: 
 What are the essential features of quantum theory that enable this spreading of classical information in the first place?
Certainly, this is possible in quantum theory's rich mathematical structure of complex Hilbert spaces,
 but can we identify a selective subset of more physically--motivated principles that similarly enable this Darwinistic emergence of classical reality?
 To approach this, we adopt the minimal--assumptions framework of generalized probabilistic theories (GPTs)~\cite{Hardy_Axioms,Barrett_GPT}.
These encompass a wide class of operational scenarios, in which a physical system is entirely characterized by its experimental statistics resulting from preparation and subsequent measurement procedures. 
The GPT approach has thus far enjoyed particular success in identifying which operational features are necessary or sufficient for quantum phenomena like teleportation~\cite{Barnum_Teleport}, no-cloning~\cite{Barnum_NoBroadcasting}, entanglement~\cite{Barrett_GPT}, phase and interference~\cite{Garner_GPTPhase,Dahlsten_BL}, or decoherence~\cite{Richens_DecoherenceGPTs}.
With this article, we aim to extend this canon to include a form of Quantum Darwinism.

We first introduce an idealized instance of Quantum Darwinism (\cref{sec:essentialQD}), which is based on CNOT-like fan-out interactions and carries essential features of Darwinism. 
After that, we provide a brief overview of the GPT framework (\cref{sec:IntroGPT}).
We then proceed to the results of the article:
 an operational formulation of idealized Quantum Darwinism (\cref{sec:definition}), followed by necessary (\cref{sec:necessary}) and sufficient (\cref{sec:sufficient})  conditions for such to exist.
Particularly, we show that both entangled states and entangled measurements are necessary features in any theory with non-classical features that exhibits Darwinism.
This suggests the counterintuitive general principle that in the presence of local non-classicality, a classical world can only emerge if this non-classicality can be ``amplified'' to a form of entanglement. 
We then identify how other physically--motivated features, such as the no-restriction hypothesis~\cite{Chiribella_Purification,Janotta_GPTsRestriction} and strong symmetry~\cite{BMU_Ix}, or the existence of decoherence~\cite{Richens_DecoherenceGPTs}, are sufficient to imply the presence of Darwinism.
Finally (\cref{sec:STM}), we give a concrete example of a non-classical theory other than quantum theory that admits Darwinism: we show its existence in Spekkens' Toy Model~\cite{Spekkens_ToyModel} and its convex extensions~\cite{Janotta_GPTsRestriction,Garner_GPTPhase}.

\section{Background}
\label{sec:Background}
\subsection{Quantum Darwinism}
\label{sec:essentialQD}

The typical setting of Quantum Darwinism (QD)~\cite{Zurek_QuantumOrigins,Zurek_QDReview1,Zurek_QDReview2} consists of a central system $\mathcal{S}$ interacting with a multi-partite environment $\mathcal{E}_1, \ldots \mathcal{E}_N$.
This is similar to the setting in which {\em decoherence} is studied (e.g.~\cite{Schlosshauer_Deco}),
 but rather than focusing on the change in $\mathcal{S}$'s state,
 QD is concerned with the information that {fragments} of the environment can learn about $\mathcal{S}$.

Not everything about $\mathcal{S}$ can be spread to the environment -- for instance, sharing arbitrary quantum information would violate the {\em no-cloning principle}~\cite{Wootters_NoCloning} {(or, more generally, the no-broadcasting theorem \cite{Barnum_NoBroadcasting})}.
Nonetheless, something can still be learned about $\mathcal{S}$ -- perhaps because the interaction induces certain quantum states on system and environment such that measurements made on $\mathcal{S}$ and $\{\mathcal{E}_{n}\}$ in the right choice of basis yield correlated outcomes.
This interaction must also preserve some aspect of the initial state of $\mathcal{S}$, so that what the environment learns can be considered as being about $\mathcal{S}$.

In the ideal scenario, we would like to extract as much classical information from any $\mathcal{E}_{n}$ about $\mathcal{S}$ as we could from $\mathcal{S}$ directly.
Holevo's theorem~\cite{Holevo_Bound} tells us that the most information that can possibly be shared with each environmental system is upper-bounded by that directly obtainable from a single measurement on $\mathcal{S}$.
A simple way by which this can be realized, when $\mathcal{S}$ and all of $\mathcal{E}_{n}$ (${n}=1\ldots N$) are $d$-dimensional quantum systems, 
 is as follows.
Let $\mathcal{M}:= \{\ket{0}, \ldots \ket{d-1}\}$ be some orthonormal basis.
Suppose $\mathcal{S}$ is initially in a pure state $\ket{\psi}_S = \sum \alpha_k \ket{k}$, 
 and each environmental system starts in a pure basis state $\ket{j_{n}}_{n} \in \mathcal{M}$. 
 (The use of $\mathcal{M}= \{\ket{0}, \ldots \ket{d-1}\}$ for the environment just means that we label the relevant environment states $\ket{0}$, $\ket{1}$, \dots, $\ket{d-1}$. It does not imply that these environment states have the same physical properties as the corresponding state of the main system or of the other environment systems.)
Consider the following {\em fan-out} gate (a generalization of control-NOT / control-shift gates, see \cref{fig:fanout}):
\begin{align}
{\rm FAN}\!\left( \ket{k}_\mathcal{S}\otimes \bigotimes_{{n}=1}^N \ket{j_{{n}}} \right) & := \ket{k}_{\mathcal{S}} \otimes \bigotimes_{{n=1}}^N \ket{j_{{n}} \oplus k},
\label{eq:QPointerStates}
\end{align}
such that
\begin{align}
 {\rm FAN}\!\left( \ket{\psi}_\mathcal{S}\otimes \bigotimes_{{n}=1}^N \ket{j_{{n}}} \right)
\hspace{-3em} & \hspace{3em} = \sum_{k} \alpha_{k} {\rm FAN}\!\left(\ket{k}_{\mathcal{S}} \otimes \bigotimes_{{n}=1}^N \ket{j_{{n}}} \right)\nonumber\\ 
& = \sum_{k} \alpha_{k} \ket{k}_{\mathcal{S}} \otimes \bigotimes_{{n}=1}^N \ket{j_{{n}} \oplus k},
\label{eq:QFanOut}
\end{align}
where $j_{n} \oplus k$ indicates addition modulo $d$.

\begin{figure}[htb]
\centering
\includegraphics[width=\linewidth]{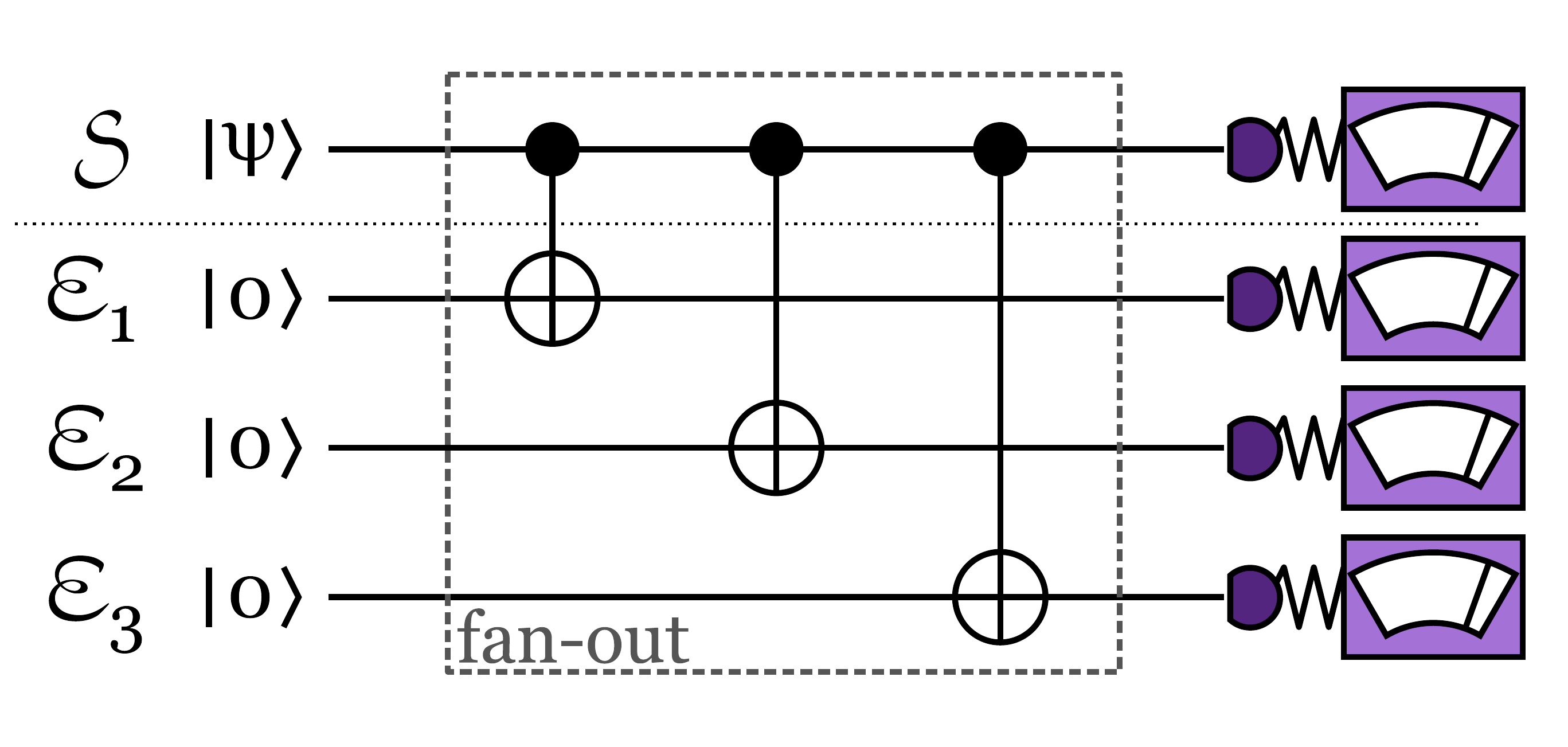}
\caption{ \caphead{{Idealized} Quantum Darwinism: fan-out gate.}
The fan-out gate (\cref{eq:QFanOut}) is realized for the case $N=3$, $d=2$ by three consecutive CNOT gates. After this process, the statistics of the computational-basis measurement $Z$ on all of the environmental subsystems ($\mathcal{E}_1$, $\mathcal{E}_2$ and $\mathcal{E}_3$) agrees with those of the main system ($\mathcal{S}$).
Meanwhile, the statistics of this measurement on the system are the same as if it had been directly made on $\ket{\psi}$.
As such, the classical information about $Z$ in $\mathcal{S}$ has been spread to its environment.
\label{fig:fanout}
}
\end{figure}

It is clear that fan-out realizes the above ideals, perfectly broadcasting classical information about $\mathcal{M}$ on $\mathcal{S}$ to every environment system, while preserving the outcome probabilities of $\mathcal{M}$ on $\mathcal{S}$.
First, if $\ket{\psi}_{\mathcal{S}} \in \mathcal{M}$ (as in Eq.~\eqref{eq:QPointerStates}), then it remains unchanged after the interaction; 
 that is, $\mathcal{M}$ is the \emph{pointer basis} selected by ${\rm FAN}$.
Moreover, if $\ket{\psi}_\mathcal{S}\in\mathcal{M}$, this element can be perfectly identified by simply measuring any of $\mathcal{E}_{n}$ with $\mathcal{M}$ and applying the appropriate relabelling (subtraction of $j_{n}$ modulo $d$) -- in which case the so-called {\em einselection process}{, as defined by}~\citet{Zurek_QuantumOrigins}, occurs on the level of each environmental subsystem.
Furthermore, when $\ket{\psi}_{\mathcal{S}}$ is a superposition of multiple {states in} $\mathcal{M}$, the resulting entangled state now has the property that whatever the outcome of $\mathcal{M}$ {on} $\mathcal{S}$,
the \emph{same} outcome will be obtained by making $\mathcal{M}$ on $\mathcal{E}_{n}$ (again, via subtraction of $j_{n}$ modulo $d$). 
Finally, the statistics of measuring $\mathcal{M}$ on $\mathcal{S}$ before and after the fan-out are identical.
Thus, such a fan-out is said to implement an {\em {idealized} Darwinism process}.

In this idealized setting implemented by fan-out, {\em any} state in {the basis} $\mathcal{M}$ represents a valid initial state of an environment $\mathcal{E}_{n}$ for which FAN can register information about $\mathcal{S}$'s pointer basis.
This multiplicity of ``good registers'' is not necessary for Darwinism: indeed, one could have a {\em minimal Darwinist process} where the emergence of classical information is observed only for specific initial environmental states (e.g.\ vacuum modes~\cite{Blume-Kohout_QDBrownian}).
However, the idealized process above is more robust to modifications in the initial state of the environment subsystems -- generically reducing `misalignment', in the language of \citet{Zwolak_NoisyChannels}.

QD can also be studied in specific physical scenarios. 
For instance, it has been analyzed in models with spin-like interactions, both theoretically \cite{Blume-Kohout_RedundantInfo,Zwolak_NoisyChannels,Zwolak_MixedEnvironment,Zwolak_Amplification} and experimentally \cite{Unden_QDNV,Ciampini_ExperimentalQD}, in models of quantum Brownian motion \cite{Blume-Kohout_QDBrownian}
 and models of a system immersed in thermal illumination \cite{Riedel_QDIllumination,Riedel_QDIlluminationExtended}, to name a few. 
In these cases, only partial information is spread (typically quantified through mutual--information quantities -- though the efficacy of this is of debate~\cite{Horodecki_NoMI, Le_ObjectivityQDxSSB}).
For instance, pointer states may not be perfectly robust to interaction, the information may not be perfectly registered in the environment \cite{Blume-Kohout_QDBrownian,Zwolak_MixedEnvironment,Riedel_QDIllumination}, or a more general class of measurements than projection onto the pointer basis may be used~\cite{Brandao_QD,Knott_QDInfiniteDimension}. 
However, for the purpose of this article (and preempting the need to cast the scenario in the operational language of GPTs, which certainly profits from simplicity), we restrict our discussion to the idealized form of Darwinism described above.

\subsection{The GPT framework}
\label{sec:IntroGPT}
Generalized probabilistic theories (GPTs) are a minimal--assumption framework in which a physical theory is specified by the statistics of every experiment that could be conducted within it.
The fundamental elements of a GPT correspond to laboratory operations,
 such as state preparations, and measurement outcomes.
In addition to the aforementioned isolation of quantum features~\cite{Barnum_Teleport,Barnum_NoBroadcasting,Barrett_GPT,Garner_GPTPhase,Dahlsten_BL,Richens_DecoherenceGPTs},
 this broad operational approach makes the GPT framework well-suited for attempts to reconstruct quantum theory either from experimental data~\cite{MazurekPRS_Reconstruction} or from sets of reasonable physically--motivated axioms~\cite{Hardy_Axioms,Chiribella_Axioms,Masanes_DerivQT,BMU_Ix}.
Theories such as quantum theory (QT) and classical probability theory (CPT) are GPTs,
 but the framework also admits more exotic theories such as ``boxworld''~\cite{Barrett_GPT} or higher-dimensional Bloch ball state spaces~\cite{Mueller_3DSpace}.

In this section, we briefly review the aspects of the GPT framework that are relevant for our discussion.
For more detailed and pedagogical introductions to the GPT framework, see e.g.~\cite{Hardy_Axioms,Barrett_GPT,Mueller_LesHouches}.

\subsubsection{Single Systems}
\label{sec: GPTsingle}
The primitive elements of a GPT are the {\em states} that one can prepare, and the outcomes of measurements (known as {\em effects}) that one can make on a given physical system. 
Mathematically, states (including sub- and super-normalised states) are given by the elements of a closed subset $A_+$ of some finite-dimensional real vector space $A$. 
With a slight abuse of notation, the physical system will also be denoted $A$. 
This subset $A_+$ is assumed to be a \emph{cone}, meaning that $\varphi,\omega\in A_+$ and $\lambda\geq 0$ imply that $\lambda\varphi\in A_+$ and $\varphi+\omega\in A_+$. 
Furthermore, $A_+$ is assumed to be \emph{generating}, i.e.\ ${\rm span}(A_+)=A$, and \emph{pointed}, i.e.\ $A_+\cap (-A_+)=\{0\}$. 
(For the example of QT, this is the cone of positive semidefinite matrices, see Example~\ref{ex:QT} below.)

Effects correspond to elements in a generating cone $E_A \subseteq A^*$, where $A^*$ is $A$'s dual space of linear functionals. 
The probability of observing effect $e\in E_A$ given a preparation $\omega \in A_+$ is given by $e(\omega)$. 
Since this must be non-negative, we must have $E_A\subseteq A_+^*$, where $A_+^* :=\{e\in A^*\,\,|\,\, e(\omega)\geq 0\mbox{ for all }\omega\in A_+\}$ is the \emph{dual cone} of $A_+$~\cite{Webster94}.
We assume the existence of a distinguished {\em unit effect} $u_A\in E_A$ such that for all $a\in E_A$ there is some $\lambda>0$ with $a \leq \lambda u_A$ (where $a \leq b$ if and only if there exists some $c \in E_A$ such that $a + c = b$).
The {\em measurements} of a theory correspond to collections of effects $\{e_i\}_{i=1\ldots N}$ that sum to $u_A$ -- each constituent effect corresponds to one mutually exclusive outcome. 
Since $\sum_i e_i(\omega)=u_A(\omega)$, we can interpret $u_A(\omega)$ as the normalization of the state $\omega$ --- that is, the total probability to obtain any outcome if the measurement is performed on the corresponding physical system.
We say an effect is {\em valid} if it can be part of a measurement (i.e. $e\in E_A$ and $e\leq u_A$). 

If $E_A=A_+^*$, we say that the system is \emph{unrestricted}, or that it \emph{satisfies the no-restriction hypothesis}~\cite{Chiribella_Purification,Janotta_GPTsRestriction}. 
From $u_A$ and $A_+$, one can infer its compact convex set of {\em normalized states} $\Omega_A := \{\omega \in A_+\, |\, u_A(\omega) = 1 \} \subset A_+$. 
An example is sketched in~\cref{fig:GPTSketch}.

\begin{figure}[bth]
\centering
\includegraphics[width=0.8\linewidth]{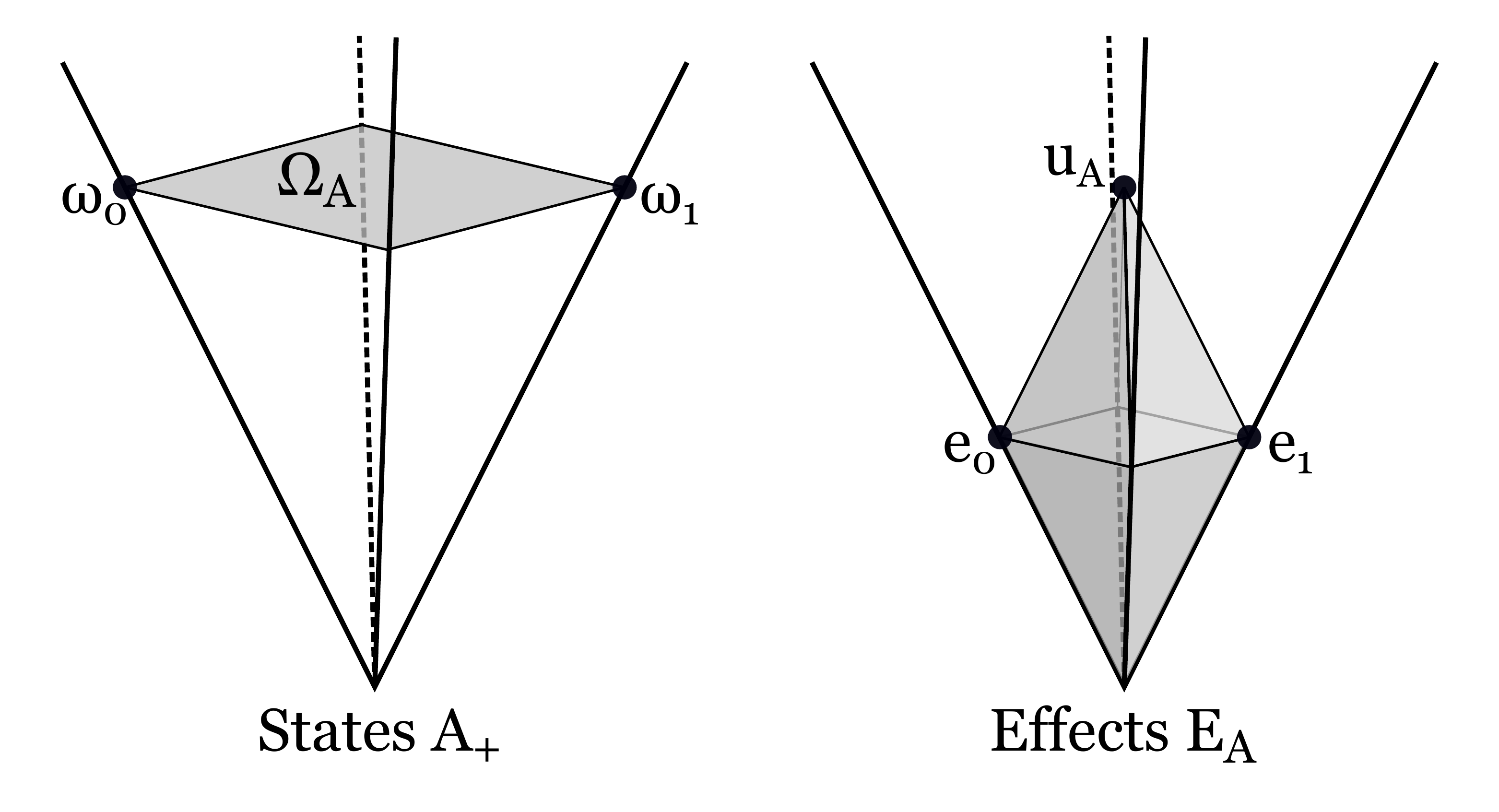}
\caption{ \caphead{Geometric picture of a GPT.}
An example state space $A_+$ (LHS) and effect space $E_A$ (RHS) of a GPT with $A=\reals^3$ is drawn.
On the RHS, the {\em unit effect} $u_A$ is labeled, and all effects on or within the shaded octahedron are {\em valid} in that $e\leq u_A$.
Two {\em pure} effects $\{e_0,e_1\}$ that satisfy $e_0\!+\!e_1\!=\!u_A$ and hence form a {\em refined measurement} are labelled.
On the LHS, the convex set of normalized states $u_A(\omega) = 1$ is shaded as $\Omega_A$.
Within it, a {\em maximal frame} of two pure states $\{\omega_0,\omega_1\}$ is labelled.
\label{fig:GPTSketch}
}
\end{figure}

The convexity of $A_+$ and $E_A$ amounts to the assumption that statistical fluctuations can always be introduced into an experiment.
Consider measurement outcome $e\in E_A$ on one of two preparations $\omega_1$ or $\omega_2$, with respective statistics $e(\omega_1)$ and $e(\omega_2)$.
If $\omega_1$ is prepared with probability $p$ and $\omega_2$ otherwise,
 then this preparation procedure should be representable by the single state $\omega$ whose statistics satisfy:
 $e(\omega) = p e(\omega_1) + (1-p)e(\omega_2) = e(p\omega_1 + (1-p)\omega_2)$.
It then follows that $\omega = p\omega_1 + (1-p)\omega_2$. A similar interpretation of convex combinations applies to the effects.

An effect $e\in E_A$ is said to be pure~\cite{Chiribella_Axioms,Chiribella_PP} {(also called ray-extremal in some literature, e.g. \cite{Mueller_Bit-Symmetry,Al_Safi_ReversibleDichotomyOddN})} if $e=\sum f_i$, with $f_i \in E_A$, implies that $f_i\propto e$ for all $i$ (see also~\cref{fig:GPTSketch}).
Pure effects cannot be obtained from (non-trivial) coarse-graining of other effects. 
A collection of pure effects that sum to $u_A$ with no effects proportional to any other in the set is known as a {\em refined measurement}. 
A {\em pure state} is defined to be a normalized state that is extremal in $\Omega_A$, i.e.\ that cannot be written as a non-trivial convex combination of other normalized states.
A \emph{frame} is a collection of pure states $\{\omega_j\}$ that can be perfectly distinguished in a single measurement: i.e.\ there is at least one measurement $\{e_i\}$ such that $e_i(\omega_j)=\delta_{ij}$. 
A \emph{maximal frame} is a frame with the largest number of distinguishable states {for that system.} 

Dynamics of single systems in the framework of GPTs are described by linear maps $T:A\to A$ known as {\em transformations}. 
Transformations $T$ must map states to states, i.e.\ $T(A_+)\subseteq A_+$, and effects to effects, in the sense that if $e\in E_A$ is a valid effect, then $e\circ T$ must also be a valid effect.  (The latter corresponds to an outcome where transformation $T$ has been applied before the measurement.) Motivated by the intuition to consider only closed-system dynamics in which all environments are explicitly modelled, we will in the following restrict our attention to \emph{reversible transformations}\footnote{{The GPT framework can also be formulated to admit discussion of irreversible transformations \cite{Barrett_GPT}.}}. These are transformations $T$ that are invertible as a linear map and whose inverse $T^{-1}$ is also a transformation. Since transformations can be composed, it follows that the reversible transformations of any GPT system $A$ form a group $\mathcal{T}_A$. Furthermore, they map the set $\Omega_A$ of normalized states onto itself.

In summary, a GPT system $A$ is defined by a tuple $(A,A_+,E_A,u_A,\mathcal{T}_A)$ of a real vector space, the state and effect cones, the unit effect, and the group of reversible transformations. Let us illustrate this framework with two familiar examples:
\begin{example}[Quantum theory (QT)] 
\label{ex:QT}
An $n$-level quantum system corresponds to the GPT system $(A^{(n)},A_+^{(n)},E_A^{(n)},u_A^{(n)},\mathcal{T}_A^{(n)})$ with
\begin{align}
&A^{(n)}=\mathbf{H}_n(\mathbb{C}),\quad A_+^{(n)}=\mathbf{H}_n^+(\mathbb{C})\simeq E_A^{(n)},\quad u_A^{(n)}=\mathbf{1}_n,\nonumber\\
&\qquad \mathcal{T}_A^{(n)}=\{\rho\mapsto U\rho U^\dagger\,\,|\,\, U^\dagger U=\mathbf{1}_n\},
\end{align}
where $\Csa{n}$ is the real vector space of $n\!\times\!n$ complex Hermitian matrices, and $\Csap{n}$ the subset of positive semidefinite matrices. Via the Hilbert-Schmidt inner product, $\langle X,Y\rangle:={\rm tr}(XY)$, we can identify $A^{(n)}$ with its dual space such that the effects are also Hermitian matrices. For example, $u_A^{(n)}(\rho)={\rm tr}(\rho)$ can be written $\langle\mathbf{1}_n,\rho\rangle$, hence we can identify $u_A^{(n)}=\mathbf{1}_n$.

The measurements $\{E_i\}_{i=1,\ldots,N}$ thus correspond to POVMs (positive operator-valued measures), i.e.\ $E_i\geq 0$ and $\sum_i E_i=\mathbf{1}_n$. The normalized states $\Omega_A^{(n)}$ are the (unit-trace) density matrices,  and the maximal frames correspond to the various $n$-element orthonormal bases of the Hilbert space $\mathbb{C}^n$. The reversible transformations are the unitary conjugations. Pure effects correspond to rank-$1$ POVM elements.
\end{example}

\begin{example}[Classical probability theory (CPT)] 
\label{ex:CPT}
A classical random variable that can take $n$ different values corresponds to the GPT system $(B^{(n)},B_+^{(n)},E_B^{(n)},u_B^{(n)},\mathcal{T}_B^{(n)})$ with
\begin{align}
&B^{(n)}=\mathbb{R}^n,\quad B_+^{(n)}=\{x\in\mathbb{R}^n\,\,|\,\,\mbox{all }x_i\geq 0\}\simeq E_B^{(n)},\nonumber\\
&\qquad u_B^{(n)}=(1,1,\ldots,1)\trans,\quad \mathcal{T}_B^{(n)}\simeq S_n.
\end{align}
In this notation, we have identified $\mathbb{R}^n$ with its dual space via the usual dot product $x\cdot y=\sum_i x_i y_i$. The unit effect is thus $u_B^{(n)}\cdot p=\sum_{i=1}^n p_i$, and so $\Omega_B^{(n)}$ is the simplex of $n$-dimensional probability vectors, i.e.\ $\Omega_B^{(n)}=\{p\in\mathbb{R}^n\,\,|\,\,p_i\geq 0,\,\sum_i {p_i} =1\}$.
The reversible transformations are the permutations of the entries: $p_i\mapsto p_{\pi(i)}$, with $\pi$ some permutation of $\{1,2,\ldots,n\}$. Thus, the group of reversible transformations is a representation of the permutation group $S_n$.

A crucial signature of classicality is that CPT has (up to relabelling) only a \emph{single} refined measurement $\{e_i\}$. 
Its effects are $e_i(p)=p_i$, and it can be interpreted as asking which of the $n$ possible configurations is actually the case. 
It distinguishes the (up to relabelling) unique maximal frame $\{\omega_j\}$, where $\omega_j:=(0,\ldots,0,\underbrace{1}_j,0,\ldots,0)\trans$. 
\end{example}
Both QT and CPT are \emph{unrestricted} and \emph{self-dual}~\cite{Mueller_Bit-Symmetry}, i.e.\ there is some inner product according to which $A_+$ is isomorphic to $E_A$. Note that GPTs will in general satisfy neither of these two properties.

\subsubsection{Maximal classical information (MCI) frames}
\label{sec:MCI}
Our goal is to generalize idealized Quantum Darwinism -- in particular, the mechanism for perfect spreading of classical information via fan-out gates -- to GPTs. 
As a first step, we have to identify the analogue of the pointer states and the measurements that read out their encoded classical information (pointer observable). 
We will focus on Darwinism generalizations that allow one to extract the maximal amount of classical information. 
This keeps the problem technically tractable in general theories, while still allowing us to explore a generalized phenomenology of quantum Darwinism.
In the quantum case, such maximal classical information is encoded onto an orthonormal basis $\{\ket{j}\}$. 
The natural analogue of this in a GPT is a maximal frame $\{\omega_j\}$.

Let us consider the measurements that could extract this classical information, i.e.\ the pointer observable. 
As seen in Example~\ref{ex:QT}, QT enjoys a strong form of duality that allows one to treat the pure states $\omega_j = \ket{j}\bra{j}$ and the corresponding rank-1 projective measurements $e_j(\bullet) = \mathrm{Tr}[\ket{j}\bra{j}\ \bullet ]$ as the ``same'' objects, and it is exactly this dual set of rank-$1$ projectors that form the measurement that extracts the maximal amount of information out of the system.

In general, GPTs do not have such an automatic duality between states and effects.
Moreover, measurements that distinguish the elements of a maximal frame do not even need to be {pure}. 
However, since we are interested in the idealized case, where one spreads the maximal classical information contained in some system, we will here focus on maximal frames that can be distinguished by a refined measurement:
\begin{definition}[Maximal classical information in GPTs]
\label{def:MCIframe}
A maximal frame $\omega_1,...,\omega_n$ is called a \textbf{maximal classical information frame} (MCI-frame) if there is a refined measurement $\{e_i\}\subset E_A$ 
which discriminates the states $\omega_j$, i.e. $e_j(\omega_k) = \delta_{jk}$.
\end{definition}
In other words, an MCI-frame is a maximal frame with the additional condition that it is discriminated by \emph{pure} effects.
Many GPTs contain MCI-frames: quantum theory certainly does (in the form of orthonormal bases), and so do quantum theory over the real numbers and over the quaternions, and $d$-ball state spaces. 
As expected, classical theories in all dimensions also have MCI-frames.
Furthermore, so-called ``dichotomic'' systems as defined in Ref.~\cite{Al_Safi_ReversibleDichotomyOddN} contain MCI-frames, which includes unrestricted systems whose sets of normalized states are regular polygons with an even number of vertices, or a $d$-cube or $d$-octoplex for $d\geq 3$.

In \cref{app:Examples}, we give an example of a state space that does \emph{not} have an MCI-frame: the pentagon. 
This example illustrates the counterintuitive properties of such systems: the pentagon has at most \emph{two} perfectly distinguishable states, but one can in some sense encode \emph{more} than one bit of information into such a system \cite{Massar_InfoPolygon}. 
That is, any classical bit that sits inside this state space does not represent the maximal amount of information that can be encoded into the system. 
For the remainder of this work, we will thus exclude such systems and focus on state spaces that contain MCI-frames.

\subsubsection{Composite systems}
\label{sec: GPTComp}
Darwinism is inherently linked to composition of subsystems; therefore, we need to understand how to treat composition in GPTs. 
There are several approaches to this~\cite{Janotta_GPT,Chiribella_Purification}, including category-theoretic formulations~\cite{Coecke_CP}. 
Here, we will motivate and state a list of minimal assumptions on any composite state space that allows us to formulate a generalization of our idealized Darwinism. 
We begin by discussing a bipartite GPT system $AB$, and then turn to the case of composing more than two subsystems.

First, we demand that the combined state space $AB$ has a notion of independent parallel preparation. 
This means that given some state $\varphi^A$ on $A$ and some state $\omega^B$ on $B$, there should be a state of $AB$ (denoted $\varphi^A\odot\omega^B$) that represents the state obtained by the independent local preparation of the two states on $A$ and on $B$.
Since statistical mixtures of local preparations must lead to statistical mixtures of the corresponding global state, the map $\odot$ must be bilinear.

As pure states can be interpreted as states of maximal knowledge, we assume that independent parallel preparations of pure states lead to global pure states, i.e.\ we demand that if $\varphi^A$ and $\omega^B$ are pure then so is $\varphi^A\odot\omega^B$~\cite{Chiribella_PP}.
Likewise, there should exist a notion of parallel implementation of measurements on the systems. 
For this we require another bilinear function (also denoted by $\odot$) that maps effects $e^A\in E_A$, $f^B\in E_B$ to effects $e^A\odot f^B\in E_{AB}$.
Furthermore, if one performs a parallel implementation of two local measurements on a composite state whose parts were prepared independently in parallel, then the probabilities should factorize in the sense that $e^A_j\odot f^B_k(\varphi^A\odot \omega^B)=e^A_j(\varphi^A)f^B_k(\omega^B)$. 
In other words, independent local procedures lead to statistical independence. 
The bilinearity of $\odot$ ensures the validity of the no-signalling principle: the choice of local measurement on $B$ does not affect the outcome probabilities of local measurements in $A$ (and vice-versa). Indeed, $\sum_j e^A_i\odot e^B_j=e_i^A \odot u^B$, for all effects $e_i^A$ and any measurement $\{e^B_j\}$. 
Similarly as for states, we assume that the composition of pure effects results in a pure effect.

Finally, we must ensure that the global structure is consistent with the local structure.
Consider a valid composite effect $e^{AB} \in E^{AB}$ and a normalized state $\omega^B \in \Omega_B$. 
Then the effect $\tilde e^A$ defined by $\tilde e^A(\varphi^A):=e^{AB}(\varphi^A\odot\omega^B)$ should be a valid effect on $A$, i.e.\ $\tilde e^A\in E_A$ {and $\tilde e^A \leq u_A$} : it can be implemented by preparing $\omega^B$ on $B$ and then measuring $e^{AB}$ on $AB$.

Similarly, consider a global state $\omega^{AB} \in \Omega_{AB}$ shared between two parties $A$ and $B$. 
Imagine that one of the parties, say $B$, implements a local measurement $\{f^B_k\}_k \subset E_B$ and tells the other party the outcome $k$. 
Then agent $A$ holds a conditional state, which should be a (subnormalized) element of the state space of $A$. 
More specifically, for an effect $e^A \in E_A$, the probability for both $f^B_k$ and $e^A$ to be obtained is given by $e^A\odot f^B_k(\omega_{AB})$. 
This implicitly defines a subnormalized state $\tilde \omega^A$ on $A$ via $e^A(\tilde\omega^A) = e^A\odot f^B_k(\omega^{AB})$, which must thus be an element of $A_+$. In the special case of the trivial measurement $f^B_k=u^B$, the state $\tilde\omega^A$ becomes the reduced state on $A$. A similar condition should hold if the roles of $A$ and $B$ are interchanged.

Together, we will call these assumptions the \emph{minimal assumptions on composition}\footnote{We do not here require any particular form of parallel composition of {\em transformations}, so we omit this from our definition.}.

\begin{definition}[Minimal assumptions on composition] 
\label{def:Composition}
A composition of GPT systems $A$ and $B$ is a GPT system $AB$ together with two bilinear maps $A\times B\to AB$ and $A^*\times B^*\to (AB)^*$, both denoted by $\odot$, satisfying the following:
\begin{enumerate}[i.]
    \item All product states are allowed and normalized: if $\omega^A\in\Omega_A$ and $\omega^B\in\Omega_B$ then $\omega^A\odot\omega^B\in\Omega_{AB}$.
    \item All products of valid effects are valid effects: if $e^A\in E_A$ and $e^B\in E_B$ then $e^A\odot e^B\in E_{AB}$.
    In particular, local measurements cannot lead to probabilities larger than $1$: $u^A\odot u^B \le u^{AB}$.
    \item Local measurements on product states yield statistically independent outcomes: $e^A\odot f^B(\omega^A\odot \omega^B)=e^A(\omega^A)f^B(\omega^B)$.
    \item Products of pure states (effects) are pure states (effects).
    \item Conditional effects: for all effects $e^{AB} \in E_{AB}$ and all normalized states $\varphi^A \in \Omega_A$ and $\omega^B\in\Omega_B$, also $e^{AB}(\varphi^A\odot \bullet)\in E_B$ and $e^{AB}(\bullet\odot \omega^B) \in E_A$ are effects.
    \item Conditional states: for all states $\omega^{AB}\in\Omega_{AB}$ and all effects $e^A\in E_A$, $f^B\in E_B$, the vectors $\tilde\omega_A$, $\tilde\omega_B$ which are implicitly defined via
    \begin{align}
    	\tilde e^A(\tilde\omega^A)&= \tilde e^A\odot f^B(\omega^{AB}) \\
    	\tilde f^B(\tilde\omega^B)&= e^A\odot \tilde f^B(\omega^{AB})
    \end{align}
     must be states, i.e.\ $\tilde\omega_A\in A_+$, $\tilde\omega_B\in B_+$.
     \item {Multipartite systems: Composites $A=A_1 A_2\ldots A_N$ of $N\geq 3$ GPT systems are defined by iterated pairwise composition (satisfying all desiderate above), in an arbitrary, but fixed order. For example, $A=(A_1 A_2) A_3$ defines $A$ as a composition of GPT systems $B$ and $A_3$, where $B$ is a composition of GPT systems $A_1$ and $A_2$. Another possibility would be to define a composite $A=A_1(A_2 A_3)$, and we are \emph{not} demanding that the resulting GPT systems must agree. We allow \emph{any} such choice as long as it is fixed within each theorem of this paper, and we will drop the brackets from the notation.
}
\end{enumerate}
\end{definition}
While these assumptions imply the no-signalling principle, we do \emph{not} demand the popular principle of ``tomographic locality''~\cite{Hardy_Axioms}, i.e.\ that the $\omega^A\odot \omega^B$ span all of $AB$. 
Thus, the $\odot$ operation cannot in general be identified mathematically with the tensor product operation. The above minimal assumptions are also compatible, for example, with QT over the real numbers~\cite{Hardy_Axioms}.

As we know from QT, a striking feature of composite systems in non-classical theories is entanglement. 
Having a definition of composite systems at hand, we are in place to define entangled states and effects in GPTs~\cite{Barnum_Teleport}:
\begin{definition}[Entangled states]
\label{def:EntangledStates}
Consider a composite system $A=A_1 A_2\ldots A_N$. States $\omega^A\in\Omega_A$ which can be written as 
\begin{align}
\omega^A = \sum_i p_i \,\omega_i^{A_1}\odot\omega_i^{A_2}\odot\ldots\odot \omega_i^{A_N}
\end{align}
with $\omega^{A_i}\in\Omega_{A_i}$ and $\{p_i\}$ a probability distribution, are called \emph{separable}. States which cannot be written in this form are called \emph{entangled}.
\end{definition}
\begin{definition}[Entangled effects]
\label{def:EntangledEffects}
Effects $e^A\in E_A$  which can be written as
\begin{align}
   e^A=\sum_i e_i^{A_1}\odot e_i^{A_2}\odot\ldots\odot e_i^{A_N}
\end{align}
with $e_i^{A_j}\in E_{A_j}$ are called \emph{separable}. Effects which cannot be written in this form are called \emph{entangled}.
\end{definition}
{
A pure {state/}effect is separable if and only if it is a product of pure {states/}effects (see, e.g. \cref{app:SepEffect}).}

\begin{shaded}
\textbf{Summary of assumptions.} 
We consider theories that satisfy:
\begin{itemize}
\item For a single system $(A, A_+, E_A, u_A,\mathcal{T}_A)$: 
 $A$ is finite-dimensional. We do {\em not} assume the no-restriction hypothesis.
\item For pairs of systems: composition satisfies all conditions of~\cref{def:Composition}. In particular,
products of pure states (or effects) are pure, but we do {\em not} assume tomographic locality.
\item For three or more systems: composites like {$ABCD$ are defined via iterated pairwise composition in an arbitrary, but fixed order (for example $(AB)(CD)$ or $A(B(CD))$). However, we do \emph{not} demand associativity.}
\end{itemize}
Unless otherwise stated, all introduced states are normalized, and all introduced effects are valid ($e\leq u_A$).
\end{shaded}

\clearpage

\section{Results}
\label{sec:results}

\subsection{A definition of Darwinism in GPTs}
\label{sec:definition}

With all these ingredients we can now ask if the fan-out mechanism for Darwinism is present in GPTs other than quantum theory.
To answer this, we must first formulate the features of the idealized Darwinism process in an {\em operational} way -- that is, in terms of experimental statistics.

To this end, recall the scenario of {idealized} \emph{Quantum} Darwinism (\cref{sec:essentialQD}).
The desire is to broadcast some classical information encoded within $\mathcal{S}$ to the environment, say, relating to pointer measurement $\mathcal{M}:= \{\ket{k}\bra{k}\}_{k=0\ldots d-1}$.
Let each environment system begin in an eigenstate of $\mathcal{M}$ (for system ${n}$, labeled $\ket{j_{n}}$),
 then after the fan-out operation ${\rm FAN}$ (\cref{eq:QFanOut}), the outcome probabilities when measuring $\mathcal{M}$ on any environment $\mathcal{E}_{n}$ will satisfy
\begin{align}
{\rm P}_\mathcal{S}(\mathcal{M}=k) = |\alpha_k|^2 = {\rm P}_{\mathcal{E}_{n}}(\mathcal{M}=j_{n}+k)  \enspace\enspace \mbox{for all }k.
\label{eq:ProbBroadcastSE}
\end{align}
Moreover, if one makes the joint measurement $\mathcal{M}^{\otimes (N+1)} = \{\ket{k_0}\bra{k_0}\otimes\dots\otimes\ket{k_N}\bra{k_N}\}$ on the entire composite system, 
  the probability of outcome $(k_0,...,k_N)$ is
\begin{align}
\label{eq:correlationsDarwinism}
p(k_0,\dots k_N) = |\alpha_{k_0}|^2 \delta_{k_0,k_1-j_1} \dots \delta_{k_0,k_N-j_N}.
\end{align}
This is the sense in which objectivity can emerge under Quantum Darwinism: 
 when this mechanism succeeds,
 all independent observers can learn about the same (maximal) classical information and agree about their findings.
Moreover, ${\rm P}_{\mathcal{S}}(\mathcal{M}\!=\!k)$ is the same before and after the fan-out is performed.

To generalize {such fan-out implemented} Darwinism to the GPT framework, we must capture the same operational behaviour on the level of probabilities.
First, we need an analogue of pointer states -- a set of distinguishable states corresponding to the classical information to be broadcast.
As mentioned in \cref{sec:MCI}, this role is played by an MCI-frame $\{\omega^{(0)}_j\}_{j=0, \ldots d-1}$ of $\mathcal{S}$ and its corresponding refined measurement $\{e^{(0)}_k\}$ with the distinguishing property $e^{(0)}_k (\omega^{(0)}_j) = \delta_{jk}$. 
Again,
 we assume the main system is in some pure state $\nu$, that may not be an element of $\{\omega^{(0)}_j\}_{j=0, \ldots d-1}$.
Lacking the mathematical structure of a Hilbert space,
 we cannot so easily express $\nu$ as a superposition of frame elements.
Nonetheless, we may readily recover the outcome probabilities when $\nu$ is measured by $\mathcal{M} := \{e^{(0)}_k\}$:
\begin{equation}
\label{eq:EasyEffect}
{\rm P}_\mathcal{S}\!\left(\mathcal{M}=k\right) = e^{(0)}_k(\nu).
\end{equation}
In the special case when $\nu$ {\em is} a member of the MCI-frame, $\nu=\omega_{j_0}^{(0)}$, we have ${\rm P}_\mathcal{S}\!\left(\mathcal{M}=k\right) = \delta_{j_0 k}$.

To carry the $d$ outcomes of the MCI-frame measurement spread from system $\mathcal{S}$, we assume that each environment system (labeled by ${n} \in \{1 \ldots N\}$) contains an MCI-frame $\{\omega^{({n})}_j\}_{j=0, \ldots d-1}$, distinguished by some refined measurement $\{e^{({n})}_k\}$. Like qubits in quantum theory, the $\mathcal{E}_{n}$ are not necessarily standalone systems like single particles, but they can correspond to effective subsystems of larger environmental systems, picked out by the specific form of the interaction with $\mathcal{S}$. Let us briefly consider the simplest case with just a single environment, initially in the first state $\omega_0$ of the frame $\{\omega^{(1)}_{j}\}$.
Then, to exhibit the same operational behaviour as \cref{eq:correlationsDarwinism} (via \cref{eq:EasyEffect}), the joint probability of any pair of outcomes $j_0$ and $j_1$ on $\mathcal{S}$ and $\mathcal{E}_1$ should satisfy
\begin{equation}
\label{eq:GeneralCNOT}
(e^{(0)}_{j_0} \odot e^{(1)}_{j_1})\left[{T_{\mathrm{ID}}}(\nu\odot\omega_0)\right] = e^{(0)}_{j_0}(\nu) \delta_{j_0,j_1}.
\end{equation}
{Here we use the label ID for ``idealized Darwinism'', which we will formally define below.}

In this way, the distribution $\{e_{j_0}(\nu)\}_{j_0}$ is broadcast to the environment, as in Eq.~\eqref{eq:ProbBroadcastSE}.
Crucially, Eq.~\eqref{eq:GeneralCNOT} implies that the system and environment will agree on the outcome of $\mathcal{M}$ on $\mathcal{S}$. 
Moreover, {the} probabilities of such an outcome when directly measuring $\mathcal{S}$ are not affected by the transformation $T_{{\rm ID}}$.
That is, $T_{{\rm ID}}$ is a member of the phase group~\cite{Garner_GPTPhase} of measurement $\{e_{j_0}\odot u\}_{j_0}$, where $u$ is the unit effect.
This can be seen by summing Eq.~\eqref{eq:GeneralCNOT} over $j_1$.

The same operational desiderata easily extend to the more general case of $N$ environmental systems, each now starting in an arbitrary frame state $\omega^{({n})}_{k_{n}}$. We summarize this with the following definitions:
\begin{definition}
\label{def:DarwinGPT}
A composition of GPT system $\mathcal{S}$ and environments $\mathcal{E}_1, \ldots, \mathcal{E}_N$ is said to admit an {\bf {idealized} Darwinism process} if
\begin{enumerate}[(a)]
\item $\mathcal{S}$ has a $d$--state MCI-frame $\{\omega^{(0)}_k\}$, discriminated by a refined measurement $\{e^{(0)}_j\}$, and
\item each $\mathcal{E}_{n}$ has a $d$-state MCI-frame $\{\omega^{({n})}_j\}$ discriminated by a refined measurement $\{e^{({n})}_j\}$, such that
\item there exists a reversible (``fan-out{-like}'') transformation $T_{{\mathrm{ID}}}\in\mathcal{T}_{\mathcal{S}\mathcal{E}_1\ldots\mathcal{E}_N}$ that satisfies
    \begin{align}
	    (e^{(0)}_{j_0}\odot e^{(1)}_{j_1}\odot... \odot e^{(N)}_{j_N})[T_{{\rm ID}}(\nu \odot \omega_{k_1}^{(1)} \odot... \odot \omega_{k_N}^{(N)})]\nonumber\\
	    = \delta_{j_1, j_0+k_1} ... \delta_{j_N, j_0 + k_N} e^{(0)}_{j_0}(\nu)
\label{eq:DarwinGPT} 
\end{align}
for all $k_1, \ldots k_N, j_0, j_1, \ldots j_N$ and {all} $\nu \in \Omega_{\mathcal{S}}$, where addition is modulo $d$.
\end{enumerate}
\end{definition}

\begin{definition}
\label{def:SpreadingCI}
If for a collection of MCI-frames $\{\omega^{(i)}_j\}$ that satisfy items (a) and (b) of~\cref{def:DarwinGPT}, a reversible transformation $T_{{\mathrm{SCI}}}\in\mathcal{T}_{\mathcal{S}\mathcal{E}_1\ldots\mathcal{E}_N}$ satisfies
\begin{align}
T_{{\mathrm{SCI}}}(\omega^{(0)}_{j_0} \odot \omega^{(1)}_{j_1} \odot \ldots \odot \omega^{(N)}_{j_N}) \hspace{-8em} & \nonumber\\ 
& = \omega^{(0)}_{j_0} \odot \omega^{(1)}_{j_0+j_1} \odot \ldots \odot \omega^{(N)}_{j_0+j_N}
\label{eq:DarwinismStrong}
\end{align}
{for all $j_0\ldots,j_N$}, then we say that $T_{{\mathrm{SCI}}}$ {\bf robustly spreads classical information}.
\end{definition}

\Cref{def:SpreadingCI} demands that the system and environment behave in some sense like classical information registers: if, for example, $j_1=\ldots=j_N=0$, the transformation $T_{{\mathrm{SCI}}}$ copies the classical information in $\mathcal{S}$ to the environments, directly on the level of states. 
In quantum theory, such robust spreading of classical information is sufficient for Darwinism: the pointer basis of $\mathcal{S}$ spans the system's Hilbert space, and so Eq.~\eqref{eq:DarwinismStrong} implies Eq.~\eqref{eq:DarwinGPT} due to the state vector linearity of unitary maps. 
More generally, Definitions~\ref{def:DarwinGPT} and~\ref{def:SpreadingCI} are equivalent in quantum theory, in the sense that idealized Quantum Darwinism processes are exactly those that robustly spread classical information.

However, this equivalence does not hold for arbitrary GPTs, since Eq.~\eqref{eq:DarwinismStrong} will not in general imply Eq.~\eqref{eq:DarwinGPT} {or vice-versa. More specifically, \cref{def:SpreadingCI} only considers the impact of $T_{\mathrm{SCI}}$ on the states $w_{j}^{(0)}$ of a specific MCI-frame. However, \cref{def:DarwinGPT} considers arbitrary states $\nu$. As non-quantum GPTs do not have an underlying Hilbert space structure, the action of $T_{\mathrm{SCI}}$ on the states of the specific MCI-frame will in general not fully determine its action for arbitrary $\nu$. Therefore, we do not expect \cref{def:SpreadingCI} to imply \cref{def:DarwinGPT}.}

In the opposite direction, even if~\cref{def:DarwinGPT} holds, \cref{def:SpreadingCI} can put additional constraints on both the system and the environment.
With respect to the system, one needs to consider the possibility of a $T_{\mathrm{{ID}}}$ that preserves the statistics of $\{e^{(0)}_{j}\}$ on $\mathcal{S}$, but still changes the state of $\mathcal{S}$, even if $\mathcal{S}$ is prepared in one of the frame states $\omega_j^{(0)}$. This is impossible in quantum theory, since every rank-$1$ quantum projector $E_j^{(0)}$ has a unique normalized and pure state $\omega_j^{(0)}$ that satisfies $\tr \left( E_j^{(0)}\omega_j^{(0)}\right) = 1$. However, many GPT systems (such as gbits~\cite{Barrett_GPT}) violate the analogous operational condition on MCI-frames, which can in some cases be traced back to the fact that GPTs need not obey the usual quantum uncertainty principles~\cite{Dahlsten_BL}. With respect to the environment, \cref{def:SpreadingCI} precludes the possibility, allowed in \cref{def:DarwinGPT}, that the transformation creates exotic correlations between the $\mathcal{E}_{n}$ while preserving the statistics of the product measurements $e_{j_1}^{(1)}\odot\ldots\odot e_{j_N}^{(N)}$. 

Thus, \cref{def:DarwinGPT} captures the essential features for idealized Darwinism on the observational level, while \cref{def:SpreadingCI} treats MCI-frames as classical subsystems acting as classical memory. 
As discussed above, in the more general framework of GPTs, these processes may be treated as logically independent.

Finally, we remark on the possibility of weaker definitions that, in certain conditions, still capture the essence of the emergence of classical objectivity: namely, a {\em minimal Darwinism process}, equivalent to \cref{def:DarwinGPT}, except that Eq.~\eqref{eq:DarwinGPT} is only required to hold for some particular choice of $k_1\ldots k_N$.
This amounts to a process in which measurements on the environmental systems are only guaranteed to agree about the system's state when they are initialized in a particular state.

\subsection{Necessary features for Darwinism in GPTs}
\label{sec:necessary}

In QT, the fan-out gate (\cref{eq:QFanOut}) can create entanglement whenever the system is not initialized to a pointer state.
The first main results of this paper are to show that entanglement--creation is a necessary property of {\em any} generalized {idealized} Darwinism process {in GPTs that are non-classical in a sense made precise below. In particular, for idealized Darwinism to be present, either the GPT shares some important features with CPT, or it must create entanglement. The first nonclassical signature we will focus on relates to the existence of pure states with indeterminate outcomes on refined measurements:}

\begin{definition}[quasi-classical MCI-frames]
\label{def:qcMCI}
An MCI-frame is \emph{quasi-classical} if for all effects $e_i$ of the distinguishing measurement, it holds that $e_i(\nu)\in\{0,1\}$ for \emph{all} pure states $\nu\in\Omega$.
\end{definition}

Pure measurements that
can give non-deterministic outcomes on some pure states are often considered a characteristic nonclassical feature of quantum theory.
GPTs that have at least one MCI-frame that is not quasi-classical must be non-classical (in the sense that they cannot be classical theory, as outlined in \cref{ex:CPT}).
Conversely, however, even if every MCI-frame is quasi-classical, this is no guarantee that the theory itself is classical theory: a counter-example being boxworld's gbits \cite{Barrett_GPT}. However, the presence of such quasi-classical MCI-frames is rare among non-classical GPTs, as exemplified in \cref{App:n-gonGPTs}.
We now present our first result, by showing that preventing a {fan-out--like} Darwinism process from creating entangled states implies that the GPT must have a quasi-classical MCI-frame.

\begin{theorem}
\label{theorem:DneedsEntangledStates}
Suppose that we have an ideal Darwinism process for which the fan-out{--like Darwinism} transformation $T_{{ID}}$ maps separable states to separable states. 
Then, for every pure state $\nu\in\Omega_{\mathcal{S}}$, we have $e_i^{(0)}(\nu)=0$ or $e_i^{(0)}(\nu)=1$ for all $i$. 
That is, the MCI-frame associated with $T_{\rm ID}$ is quasi-classical.
\end{theorem}

\noindent \textit{Remark.}
This conclusion is valid also for other forms of idealized Darwinism processes that, instead of \cref{def:DarwinGPT}, satisfy the weaker minimal Darwinism condition
\begin{align}
\left(e_{j_0}\odot e_{j_1}^{(1)}\odot\ldots\odot e_{j_N}^{(N)}\right)T\left(\nu\odot \omega^{(1)}\odot\ldots\odot\omega^{(N)}\right) \hspace{-15em} & \nonumber\\
& =e_{{j_0}}(\nu)\delta_{{j_0},j_1}\delta_{{j_0},j_2}\ldots \delta_{{j_0},j_N}
   \label{eq:DarwinCNOT}
\end{align}
for all states $\nu\in\Omega_{\mathcal{S}}$, where $\omega^{(1)},\ldots,\omega^{(N)}$ is {some fixed set of pure states (not necessarily an MCI frame), and the $\{e_i\}_i$ and $\{e^{(j)}_{j_i}\}_i$ are some fixed measurements (as opposed to refined measurements) and $T$ is a reversible transformation}.

\begin{proof}
Since $T$ is a reversible transformation, it maps pure states to pure states. 
Hence, if it also preserves separability, then there are pure states $\varphi^{(0)},\ldots,\varphi^{(N)}$ (which may all depend on $\nu$) such that (see \cref{app:SepEffect})
\begin{align}
   T\left(\nu\odot\omega^{(1)}\odot \ldots \odot \omega^{(N)}\right)=\varphi^{(0)}\odot\varphi^{(1)}\odot\ldots\odot\varphi^{(N)}.
\end{align}
Since $T$ satisfies Eq.~(\ref{eq:DarwinCNOT}), we obtain
\begin{equation}
	e_i(\nu)\delta_{i,j_1}\ldots\delta_{i,j_N}=e_i(\varphi^{(0)}) e_{j_1}^{(1)}(\varphi^{(1)})\ldots e_{j_N}^{(N)}(\varphi^{(N)}).
	\label{eqDeltakram}
\end{equation}
Summing over all $j_1,\ldots,j_N$ yields $e_i(\nu)=e_i(\varphi^{(0)})$ for all $i$.

Now suppose that $i^*$ is an outcome label such that $e_{i^*}(\varphi^{(0)})=0$, then $e_{i^*}(\nu)=0$. On the other hand, consider the case that $e_{i^*}(\varphi^{(0)})\neq 0$. If at least one of the $j_k$ is different from $i^*$, then setting $i=i^*$ in Eq.~(\ref{eqDeltakram}) yields
\begin{align}
   0=\underbrace{e_{i^*}(\varphi^{(0)})}_{\neq 0} e_{j_1}^{(1)}(\varphi^{(1)})\ldots e_{j_N}^{(N)}(\varphi^{(N)}),
\end{align}
hence $e_{j_1}^{(1)}(\varphi^{(1)})\ldots e_{j_N}^{(N)}(\varphi^{(N)})=0$. But since $\sum_{j_1,\ldots,j_N} e_{j_1}^{(1)}(\varphi^{(1)})\ldots e_{j_N}^{(N)}(\varphi^{(N)})=1$, we must have $e_{i^*}^{(1)}(\varphi^{(1)})\ldots e_{i^*}^{(N)}(\varphi^{(N)})=1$, and so $e_{i^*}^{(j)}(\varphi^{(j)})=1$ for all $j$. Recalling Eq.~(\ref{eqDeltakram}) we therefore see that $e_i(\nu)=0$ for all $i\neq i^*$, and so $e_{i^*}(\nu)=1$.

In summary, we obtain $e_{i^*}(\nu)\in\{0,1\}$ for all $i^*$.
\end{proof}
Thus, for all GPT systems $\mathcal{S}$ that contain pure states on which the MCI-frame is not quasi-classical, the corresponding {idealized} Darwinism processes (if they exist) must create entangled states. While this property will be satisfied for typical GPT systems, we cannot immediately conclude that a system satisfying $e(\nu)= 0$ or $1$ {for all pure $\nu$} must be classical {(recall the gbit)}. 
However, as we shall see in following theorem, Darwinism in boxworld (among a wider class of theories) can be ruled out by another necessary condition: this time, on the measurements.

\begin{definition}[Inequivalent refined measurements]
\label{def:inequivalent}
Two refined measurements $\{e_0,\ldots,e_{d-1}\}$ and $\{\tilde{e}_0,\ldots,\tilde{e}_{d-1}\}$ are said to be {\bf inequivalent} if $\{\tilde{e}_i\}$ is not just a relabelling of the measurement $\{e_j\}$, i.e.\ at least one of the $\tilde{e}_i$ is not equal to any of the $e_j$.
\end{definition}

{In CPT, there is a unique refined measurement that sums up to the order unit, which is the distinguishing measurement of the unique MCI-frame.
Therefore, the existence of an inequivalent (refined) measurement also indicates that a GPT has a nonclassical aspect.}
In quantum theory, for instance, inequivalent measurements  correspond to {rank-one} projective measurements in different bases.

\begin{theorem}
\label{thm:EntangledEffects}
{Suppose that we have an idealized Darwinism process such that the system $\mathcal{S}$ has an MCI-frame--distinguishing measurement,
 and at least one other inequivalent refined measurement with the same number of outcomes $d$.}
Then the fan-out{--like Darwinism} transformation $T_{{\rm ID}}$ must map some pure product effects to entangled effects.
\begin{proof}
It will be useful to use the notation $\|e\|:=\max_{\omega\in\Omega} e(\omega)$ for effects $e$. 
Suppose that $T_{{\rm ID}}$ maps all pure product effects to separable effects. 
Then, since $T_{{\rm ID}}$ is reversible and preserves purity, Lemma~\ref{LemPureProductEffects} {(\cref{app:SepEffect})} implies that $T_{{\rm ID}}$ maps pure product effects to pure product effects. 
Hence, due to~\cref{eq:DarwinGPT}, for every $j_0$ and for every $\mathbf{j}=(j_1,\ldots,j_N)$ there are effects $h_{j_0,\mathbf{j}}^{(0)},\ldots,h_{j_0,\mathbf{j}}^{(N)}$ such that
\begin{align}
& \left(e_{j_0}^{(0)}\odot e_{j_1}^{(1)}\odot\ldots\odot e_{j_N}^{(N)}\right)T_{{\rm ID}} \nonumber \\
& \qquad = \, h_{j_0,\mathbf{j}}^{(0)}\odot h_{j_0,\mathbf{j}}^{(1)}\odot\ldots\odot h_{j_0,\mathbf{j}}^{(N)}.
   \label{eq14}
\end{align}
Due to multilinearity, we can move any multiplicative constant into the zeroth factor, and in this way choose the effects such that $\|h_{j_0,\mathbf{j}}^{(i)}\|=1$ for all $i\in\{1,\ldots,N\}$. 
If we had $\|h_{j_0,\mathbf{j}}^{(0)}\|<1$, then the right-hand side could never attain the value $1$ on product states, but we know that it does due to {\cref{def:DarwinGPT}}.
Thus, $\|h_{j_0,\mathbf{j}}^{(0)}\|=1$.

Substituting \cref{eq14} into \cref{eq:DarwinGPT} and noting that the result is valid for every state $\nu\in\Omega_{\mathcal{S}}$, we obtain
\begin{align}
   h_{j_0,\mathbf{j}}^{(0)}\, p_{j_0,\mathbf{j},\mathbf{k}}=\delta_{j_1,j_0+k_1}\ldots \delta_{j_N,j_0+k_N} e_{j_0}^{(0)},
\end{align}
where $p_{j_0,\mathbf{j},\mathbf{k}}:=h_{j_0,\mathbf{j}}^{(1)}(\omega_{k_1}^{(1)})\cdot\ldots\cdot h_{j_0,\mathbf{j}}^{(N)}(\omega_{k_N}^{(N)})\geq 0$. 
The special case of $\mathbf{k}=\mathbf{j}-j_0:=(j_1-j_0,\ldots,j_N-j_0)$ yields $e_{j_0}^{(0)}=p_{j_0,\mathbf{j},\mathbf{j}-j_0} h_{j_0,\mathbf{j}}^{(0)}$.
But since $\|e_{j_0}^{(0)}\|=1=\|h_{j_0,\mathbf{j}}^{(0)}\|$, this implies that $h_{j_0,\mathbf{j}}^{(0)}= e_{j_0}^{(0)}$ for all $\mathbf{j}$.

Since $\mathcal{S}$ is non-classical, there is another refined measurement $\{\tilde e_j^{(0)}\}_j$ which is not just a relabelling (i.e.\ permutation) of $\{e_i^{(0)}\}_i$. 
Using again our assumption that $T_{{\rm ID}}$ maps products of pure effects to product effects, we obtain
\begin{align}
&\left(\tilde e_{j_0}^{(0)}\odot e_{j_1}^{(1)}\odot\ldots\odot e_{j_N}^{(N)}\right)T_{{\rm ID}} \nonumber \\
&\qquad = \, \tilde h_{j_0,\mathbf{j}}^{(0)}\odot \tilde h_{j_0,\mathbf{j}}^{(1)}\odot\ldots\odot\tilde h_{j_0,\mathbf{j}}^{(N)}
   \label{eqTtoTheRight}
\end{align}
for some suitable effects $\tilde h_{j_0,\mathbf{j}}^{(0)},\ldots,\tilde h_{j_0,\mathbf{j}}^{(N)}$. 
Again, we define the effects such that $\|\tilde h_{j_0,\mathbf{j}}^{(i)}\|=1$ for all $i\in\{1,\ldots,N\}$ (the case $i=0$ will be discussed later). 
Summing over $j_0$, using that $\sum_{j_0}\tilde e_{j_0}^{(0)}=u_{\mathcal{S}}=\sum_{j_0} e_{j_0}^{(0)}$, yields
\begin{align}
&\sum_{j_0} e_{j_0}^{(0)}\odot h_{j_0,\mathbf{j}}^{(1)}\odot\ldots\odot h_{j_0,\mathbf{j}}^{(N)} \nonumber \\
&\qquad = \sum_{j_0} \tilde h_{j_0,\mathbf{j}}^{(0)}\odot \tilde h_{j_0,\mathbf{j}}^{(1)}\odot\ldots\odot \tilde h_{j_0,\mathbf{j}}^{(N)}.
\label{eqInProof2}	
\end{align}
Applying both sides to the product state $\nu\odot\omega_{k_1}^{(1)}\odot\ldots\odot\omega_{k_N}^{(N)}$ and recalling Eq.~(\ref{eq:DarwinGPT}), we obtain
\begin{align}
&   \sum_{j_0} \tilde h_{j_0,\mathbf{j}}^{(0)}(\nu)\tilde h_{j_0,\mathbf{j}}^{(1)}(\omega_{k_1}^{(1)})\ldots \tilde h_{j_0,\mathbf{j}}^{(N)}(\omega_{k_N}^{(N)}) \nonumber\\
&\qquad     =\sum_{j_0} e_{j_0}^{(0)}(\nu)\delta_{j_1,j_0+k_1}\ldots\delta_{j_N,j_0+k_N}.
\end{align}
So far, $\mathbf{j}$ and $\mathbf{k}$ are arbitrary, but now set $j_i:=k_i+l$ for all $i$, where $l$ is fixed (we abbreviate this by $\mathbf{j}=\mathbf{k}+l$). We obtain
\begin{align}
   e_l^{(0)}(\nu)=\sum_{j_0}q_{j_0,\mathbf{k},l} \tilde h_{j_0,\mathbf{k}+l}^{(0)}(\nu),
\end{align}
where $q_{j_0,\mathbf{k},l}:=\tilde h_{j_0,\mathbf{k}+l}^{(1)}(\omega_{k_1}^{(1)})\ldots \tilde h_{j_0,\mathbf{k}+l}^{(N)}(\omega_{k_N}^{(N)})\in[0,1]$. 
Since this is true for all states $\nu\in\Omega_{\mathcal{S}}$, we may again drop the $\nu$ and read it as an equality between effects.
Since $e_l^{(0)}\neq 0$, for every $l$ and for every $\mathbf{k}$ there must be some $j_0$ such that $q_{j_0,\mathbf{k},l}\neq 0$. 
Since $e_l^{(0)}$ is pure, this implies that $e_l^{(0)}\propto \tilde h_{j_0,\mathbf{k}+l}^{(0)}$. Now fix an arbitrary $\mathbf{j}$, and consider the special case $\mathbf{k}:=\mathbf{j}-l$. It follows that for all $l$, there exists at least one $j_0$ such that $e_l^{(0)}$ is a scalar multiple of $\tilde h_{j_0,\mathbf{j}}^{(0)}$. There are $d$ different linearly independent $e_l^{(0)}$ (labelled by $l$), and there are $d$ different $\tilde h_{j_0,\mathbf{j}}^{(0)}$, labelled by $j_0$.
Thus, to every $l$ there is a unique $j_0$ such that $e_l^{(0)}=q_{j_0,\mathbf{j}-l,l} \tilde h_{j_0,\mathbf{j}}^{(0)}$. We have
\begin{align}
   1=\|e_l^{(0)}\|=\underbrace{q_{j_0,\mathbf{j}-l,l}}_{\leq 1} \underbrace{\|\tilde h_{j_0,\mathbf{j}}^{(0)}\|}_{\leq 1},
\end{align}
hence $\|\tilde h_{j_0,\mathbf{j}}^{(0)}\|=1$, and so $e_l^{(0)}=\tilde h_{j_0,\mathbf{j}}^{(0)}$. We can rephrase this as follows. For every $\mathbf{j}$ there is a permutation $\pi$ of the indices such that $\tilde h_{j_0,\mathbf{j}}^{(0)}=e_{\pi(j_0)}^{(0)}$ for all $j_0$.

Now fix some $\mathbf{j}$. Let us return to Eq.~(\ref{eqInProof2}) and apply it to $\omega_{\pi(j_0)}^{(0)}\odot\omega$, where $\pi$ is the permutation corresponding to $\mathbf{j}$, and $\omega$ is an arbitrary global state of the $N$ environments. Using the identities that we have just derived and $e_{j_0}^{(0)}(\omega_i^{(0)})=\delta_{j_0,i}$, we  obtain
\begin{align}
   \tilde h_{j_0,\mathbf{j}}^{(1)}\odot\ldots\odot \tilde h_{j_0,\mathbf{j}}^{(N)}=h_{\pi(j_0),\mathbf{j}}^{(1)}\odot\ldots\odot h_{\pi(j_0),\mathbf{j}}^{(N)}.
\end{align}
Recalling eqs.~(\ref{eq14}) and~(\ref{eqTtoTheRight}), it follows that
\begin{align}
\left(\tilde e_{j_0}^{(0)}\odot e_{j_1}^{(1)}\odot\ldots\odot e_{j_N}^{(N)}\right)T_{{\rm ID}} \hspace{-8.5em} & \nonumber \\
& = \tilde h_{j_0,\mathbf{j}}^{(0)}\odot \tilde h_{j_0,\mathbf{j}}^{(1)}\odot\ldots\odot \tilde h_{j_0,\mathbf{j}}^{(N)} \nonumber \\
& =e_{\pi(j_0)}^{(0)}\odot h_{\pi(j_0),\mathbf{j}}^{(1)}\odot\ldots\odot h_{\pi(j_0),\mathbf{j}}^{(N)} \nonumber \\
& =\left(e_{\pi(j_0)}^{(0)}\odot e_{j_1}^{(1)}\odot\ldots\odot e_{j_N}^{(N)}\right)T_{{\rm ID}}.
\end{align}
Since $T_{{\rm ID}}$ is reversible, the terms in the brackets must be identical. 
Consider the case $j_1=j_2=\ldots=j_N=0$ and environment states $\varphi^{(1)},\ldots,\varphi^{(N)}$ with $e_0^{(k)}(\varphi^{(k)})=1$ for all $k=1,\ldots,N$. 
Then applying the above brackets to the product state $\nu\odot\varphi^{(1)}\odot\ldots\odot\varphi^{(N)}$ yields $\tilde e_{j_0}^{(0)}(\nu)=e_{\pi(j_0)}^{(0)}(\nu)$. 
Since this is true for all $\nu\in\Omega_{\mathcal{S}}$, we obtain $\tilde e_{j_0}^{(0)}=e^{{(0)}}_{\pi(j_0)}$. This contradicts our assumption that $\{\tilde e_j^{(0)}\}_j$ is not just a permutation of the $\{e_i^{(0)}\}_i$.
\end{proof}
\end{theorem}
Thus, {in GPTs with at least one additional inequivalent refined measurement with the same number of outcomes as the MCI frame measurement,} a reversible transformation $T_{{ID}}$ that implements an {idealized} Darwinism process will create entangled effects. 
An important consequence is that GPTs without entangled effects, such as those constructed by taking the maximal tensor product in the context of tomographic locality, cannot admit such a process. 
In particular, this rules out Darwinism in boxworld~\cite{Barrett_GPT} (a theory containing the aforementioned gbits) or any dichotomic maximally nonlocal theory. For these specific examples, one could also infer this from Refs.~\cite{Gross_TBoxworld1,Al_Safi_ReversibleDichotomyOddN}, but here we have shown it without having to determine the complete structure of the reversible transformations. 

Interestingly, entanglement for states is also needed in general physical theories if one imposes another condition of relevance for the classical limit: the existence of a decoherence map~\cite{Richens_DecoherenceGPTs}. 
{However,} theories that have an {idealized} Darwinism process -- and by our results need entangled states and measurements -- may not contain such a decoherence map, as we shall show in section \ref{sec:STM}. Therefore, our results provide not only alternative proofs but are complementary to that of \citet{Richens_DecoherenceGPTs}: together they support the idea that this non-classical feature must be present for a locally non-classical theory to admit a meaningful classical limit.

\subsection{Sufficient features for Darwinism in GPTs}
\label{sec:sufficient}
Let us now determine sufficient conditions that guarantee that Quantum Darwinism can be generalized into a theory. 
In particular, we are interested in which operationally well-motivated postulates that have already appeared in the GPT literature can lead to such Darwinism. In this spirit, we will see how a framework that admits decoherence also admits Darwinism.

We will first determine sufficient structure in GPTs to allow for the robust spreading of classical information (in the manner of \cref{def:SpreadingCI}),
before determining which additional postulates can be added to guarantee the existence of an {idealized} Darwinism process (\cref{def:DarwinGPT}) that additionally broadcasts classical information to the environment even when the system is not in a MCI--frame state.

Recall that both, the spreading of classical information and the idealized Darwinism processes, require the system to have an MCI-frame (playing the role of pointer states) that defines the classical information to be spread to the environment (\cref{def:DarwinGPT}(a)).
Likewise the environments must admit MCI-frames on which to receive this classical information (\cref{def:DarwinGPT}(b)).
Even though a theory admitting such frames may arguably be said to contain classical information (i.e.\ admitting ``registers'' that can encode the appropriate values),
 it may not generally admit all (or even any)\ classical information processing -- that is, there is no guarantee that the theory admits sufficient dynamics to satisfy \cref{def:SpreadingCI}.
In the following, we will consider what physical characteristics {\em do} ensure that the theory has enough classical information processing power to implement a fan-out gate in the manner of \cref{eq:DarwinismStrong}.

The first possible characteristic is to demand that composite systems satisfy \emph{strong symmetry}~\cite{Barnum_HigherOrder}:
\begin{definition} \label{def:StrongSymmetry}
A GPT system with group of reversible transformations $\mathcal T$ satisfies \textbf{strong symmetry (on states)} if for all $n \in \mathbb N$ and for all pairs of frames {$\{\omega_1,...,\omega_n\}$} and {$\{\nu_1,...,\nu_n\}$}, there exists some $T \in \mathcal T$ with $T \omega_j = \nu_j$ for all $j$.
\end{definition}
Strong symmetry says that all ways of encoding classical information are computationally equivalent. 
In particular, it implies that classical reversible computation can be performed on the MCI-frames of system and environment: since the set of states $\omega_{j_0,...,j_N} := \omega_{j_0}^{(0)} \odot ... \odot \omega_{j_N}^{(N)}$ constitutes a frame of the composite system, strong symmetry implies that we can perform arbitrary classical reversible gates (and thus arbitrary permutations) of those frame elements. This immediately gives us the following result:
\begin{lemma}
\label{lem:StrongSymmetry}
Consider GPT systems $\mathcal{S},\mathcal{E}_1,\ldots,\mathcal{E}_N$ that carry $d$-outcome MCI-frames. 
Every composition $\mathcal{S}\mathcal{E}_1\ldots\mathcal{E}_N$ that satisfies strong symmetry (on states) admits the robust spreading of classical information.
\end{lemma}
While strong symmetry on states implies the robust spreading of classical information in the sense of \cref{def:SpreadingCI}, we do not know whether this property implies the existence of an idealized Darwinism process in the sense of \cref{def:DarwinGPT}. 
Interestingly, the existence of such a process follows if we consider a dual notion of strong symmetry on the \emph{measurements}:
\begin{definition} \label{def:StrongSymmetryEffects}
A GPT system with group of reversible transformations $\mathcal T$ satisfies \textbf{strong symmetry (on effects)} if the following holds for all $n \in \mathbb N$: If $(e_1,\ldots,e_n)$ is a collection of pure effects that perfectly distinguishes some frame, and so is $(f_1,\ldots,f_n)$, then there exists a $T\in\mathcal{T}$ with $e_j=f_j\circ T$ for all $j$.
\end{definition}
If this property holds, we can show the following:
\begin{lemma}
\label{lem:StrongSymmetryEffects}
Consider again GPT systems $\mathcal{S},\mathcal{E}_1,\ldots,\mathcal{E}_N$ that carry $d$-outcome MCI-frames. Every composition $\mathcal{S}\mathcal{E}_1\ldots\mathcal{E}_N$ that satisfies strong symmetry (on effects) admits an {idealized} Darwinism process.
\end{lemma}
\begin{proof}
The $e_{j_0,\ldots,j_N}{:= e_{j_0}^{(0)}\odot e_{j_1}^{(1)}\odot\ldots\odot e_{j_N}^{(N)}}$ are pure effects which perfectly distinguish the frame $\omega_{j_0,\ldots,j_N}$. Thus, strong symmetry on effects implies that there is some $T_{{\rm ID}}\in\mathcal{T}_{\mathcal{S}\mathcal{E}_1\ldots\mathcal{E}_N}$	 with 
\begin{equation}
   e_{j_0,j_1,\ldots,j_N}\circ T_{{\rm ID}} = e_{j_0,j_1-j_0,\ldots,j_N-j_0},
   \label{eqChooseT}
\end{equation}
where subtraction is modulo $d$. One can check directly that this map $T_{{\mathrm{ID}}}$ {indeed} satisfies Eq.~\eqref{eq:DarwinGPT}.
\end{proof}
Thus, the version of Darwinism that is guaranteed to hold (according to \cref{def:DarwinGPT} or \cref{def:SpreadingCI}) depends on whether we demand strong symmetry on the states or on the effects. Is there a way to guarantee it on both? Indeed, it turns out that the no-restriction hypothesis is sufficient for this:

\begin{theorem}
\label{TheSSImpliesD}
Consider GPT systems $\mathcal{S},\mathcal{E}_1,\ldots,\mathcal{E}_N$ that carry $d$-outcome MCI-frames. Every unrestricted composition $\mathcal{S}\mathcal{E}_1\ldots\mathcal{E}_N$ that satisfies strong symmetry (on states) has a transformation $T_{{\rm ID}}\in\mathcal{T}_{\mathcal{S}\mathcal{E}_1\ldots\mathcal{E}_N}$ that robustly spreads classical information \emph{and} that generates an {idealized} Darwinism process.
\end{theorem}

\begin{proof}
For unrestricted systems $A$ with strong symmetry on states, it was shown in Ref.~\cite{Mueller_Bit-Symmetry} that there is a particularly strong duality between states and effects: there is an inner product $\langle\cdot,\cdot\rangle$ on $A$ such that frames $\omega_1,\ldots,\omega_n$ correspond to orthonormal systems, and the corresponding pure effects with $e_i(\omega_j)=\delta_{ij}$ must be given by $e_i(\omega)=\langle \omega_i,\omega\rangle$. Moreover, all $T\in\mathcal{T}_A$ are orthogonal with respect to this inner product. If $f_1,\ldots,f_n$ is any other collection of pure effects that distinguish a frame (say, $\nu_1,\ldots,\nu_n$), then strong symmetry on states says that there is some $T\in\mathcal{T}_A$ with $T\omega_j=\nu_j$, and so
\begin{align}
e_j\circ T^{-1}(\omega)= \langle\omega_j,T^{-1}\omega\rangle=\langle T\omega_j,\omega\rangle=\langle\nu_j,\omega\rangle=f_j(\omega).
\end{align}
Consequently, $A$ also satisfies strong symmetry on effects. Now, choose $T_{{\mathrm{ID}}}$ as in Eq.~\eqref{eqChooseT}, then we already know that it generates an ideal Darwinism process. Moreover, we have just seen that $T^{-1}_{{\mathrm{ID}}}$ maps the corresponding frame elements onto each other, i.e.
\begin{align}
   T^{-1}_{{\mathrm{ID}}} \ \omega_{j_0,j_1,\ldots,j_N}=\omega_{j_0,j_1-j_0,\ldots,j_N-j_0}.
\end{align}
Applying $T_{{\mathrm{ID}}}$ to both sides shows that $T_{{\mathrm{ID}}}$ robustly spreads classical information in the sense of \cref{def:SpreadingCI}.
\end{proof}
A second path to this spreading of classical information arises from decoherence theory.
In quantum theory, decoherence plays an important role in Quantum Darwinism by explaining in some sense why we see classical probabilities  instead of superposition states.
Recently, a decoherence formalism for GPTs was developed~\cite{Richens_DecoherenceGPTs}, and we shall here see that it enables Darwinism in GPTs as well. We adapt the decoherence formalism of \citet{Richens_DecoherenceGPTs} to our setting:
\begin{definition}[Decoherence maps] 
\label{def:DecoherenceMaps}
Consider any GPT system $A$. A linear map $D:A\to A$ is called a \emph{decoherence map} if the following properties hold :
	\begin{enumerate}
		\item[1.] The image of $A_+$ under $D$ is isomorphic to a classical state space, i.e.\ there exists a frame {$\{\omega_0,...,\omega_{d-1}\}\subset\Omega_A$} such that $D(\Omega_A) = \conv\{\omega_0,...,\omega_{d-1}\}$ (i.e.\ the convex hull of the $\{\omega_i\}$). Consequently, $D$ is normalization-preserving, i.e.\ $u_A\circ D=u_A$.
		\item[2.] $D$ is idempotent, i.e.\ $D\circ D=D$.
		\item[3.] For every classical reversible transformation $T_{\rm C}:D(A)\to D(A)$ there is a reversible transformation $T\in\mathcal{T}_A$ that implements $T_C$, i.e.\ $T(\omega)=T_C(\omega)$ for all $\omega\in D(A_+)$. Not only does this map $T$ preserve the classical state space $D(A_+)$, but it also preserves the corresponding classical effect space $E_A\circ D$.
	\end{enumerate}
Furthermore, if we have a composite GPT system $A=A_1 A_2\ldots A_N$ with decoherence maps $D_1$, \ldots, $D_N$,
	\begin{enumerate}
		\item[4.] $A$ has a decoherence map $D_{1\ldots N}$ that acts as 
\begin{align}
			D_{1\ldots N}(\nu_1 \odot \ldots \odot \nu_N) = D_1(\nu_1) \odot \ldots \odot D_N(\nu_N).
\end{align}
	\end{enumerate}
\end{definition}
\Citet{Richens_DecoherenceGPTs} additionally assume that $D$ is physically implementable, but we do not assume this here.

In the following, we will need a simple property of decoherence maps:
\begin{lemma}
\label{LemCommute}
Consider a GPT system $A$ with decoherence map $D$, and $T$ any reversible transformation that implements some classical transformation in the sense of \cref{def:DecoherenceMaps} item 3. Then $DT=TD$.
\end{lemma}
\begin{proof}
Let $e\in E_A$ and $\varphi\in A_+$, then $f:=e\circ D$ is an element of the classical effect space $E_A\circ D$, and so is $f':=f\circ T$, hence $f'=f'\circ D$. Thus, we have
\begin{align}
   e\circ DT\varphi = f\circ T\varphi=f'(\varphi)=f'\circ D\varphi=e\circ DTD\varphi.
\end{align}
Since $A_+$ and $E_A$ span $A$ and $A^*$, respectively, it follows that $DT=DTD$. But $T$ preserves $D(A)={\rm span}(D(A_+))$, hence $DTD=TD$.
\end{proof}
In analogy with how quantum systems decohere to mixtures of pointer states, it is natural to consider Darwinism for frames that can result from decoherence processes.
\begin{definition}
Consider any GPT system $A$. We say that an MCI-frame $\{\omega_i\}\subset\Omega_A$ together with a corresponding refined measurement $\{e_i\}\subset E_A$ \emph{arises from decoherence} if there is a decoherence map $D:A\to A$ such that $D(A_+)={\rm cone}\{\omega_i\}$ and $E_A\circ D={\rm cone}\{e_i\}$.
\end{definition}
In this definition, ${\rm cone}\{\omega_i\}$ denotes the set of non-negative linear combinations of the $\omega_i$, i.e.\ the convex cone of unnormalized states generated by the MCI-frame (similarly for the $\{e_i\}$).

Let $\{\omega_j^{(0)}\}_{j=0}^{d-1}$ be an MCI-frame of the main system $\mathcal{S}$ that arises from decoherence map $D_0$, and similarly let $\{\omega_j^{(i)}\}_{j=0}^{d-1}$, $i=1,\ldots,N$, be MCI-frames of the environmental systems $\mathcal{E}_1,\ldots,\mathcal{E}_N$ that arise from decoherence maps $D_1,...,D_N$. Then requirement 4 of \cref{def:DecoherenceMaps} implies that there is a decoherence map $D_{0\ldots N}$ with
\begin{align}
	D_{0\ldots N}(\omega_{j_0}^{(0)} \odot \ldots \odot \omega_{j_N}^{(N)}) = D_0(\omega_{j_0}^{(0)}) \odot \ldots \odot D_N(\omega_{j_N}^{(N)}).
\end{align}
Since each $D_i$ is a projection map and since every $\omega_j^{(i)}$ is in its image, we have $D_i(\omega_j^{(i)})=\omega_j^{(i)}$, and hence
\begin{align}
	D_{0\ldots N}(\omega_{j_0}^{(0)} \odot \ldots \odot \omega_{j_N}^{(N)}) = \omega_{j_0}^{(0)} \odot \ldots \odot\omega_{j_N}^{(N)}.
\end{align}
Requirement 3 for decoherence maps implies that the classical transformation defined by Eq.~\eqref{eq:DarwinismStrong} (a particular permutation of the classical pure states) is implemented as a reversible transformation $T_{{\mathrm{SCI}}} \in\mathcal{T}_A$ on the composite GPT system $A:=\mathcal{S}\mathcal{E}_1\ldots\mathcal{E}_N$. This transformation hence robustly spreads classical information in the sense of \cref{def:SpreadingCI}.

Furthermore, consider any state $\nu\in\Omega_{\mathcal{S}}$, and let $\nu_0:=D_0\nu$. Since the MCI-frame of $\mathcal{S}$ arises from $D_0$, there is a convex decomposition $\nu_0=\sum_{i=0}^{d-1} \lambda_i \omega_i^{(0)}$ with $\lambda_i\geq 0$, $\sum_{i=0}^{d-1}\lambda_i=1$. Using Lemma~\ref{LemCommute}, we thus obtain
\begin{align}
(e^{(0)}_{j_0}\odot e^{(1)}_{j_1}\odot\ldots \odot e^{(N)}_{j_N}) T_{{\mathrm{SCI}}}  (\nu \odot \omega_{k_1}^{(1)} \odot\ldots \odot \omega_{k_N}^{(N)}) \hspace{-21.5em} & \nonumber \\
&= e_{j_0,\ldots,j_N}\circ D_{0\ldots N} T_{{\mathrm{SCI}}} (\nu \odot \omega_{k_1}^{(1)} \odot\ldots \odot \omega_{k_N}^{(N)})\nonumber \\
&= e_{j_0,\ldots,j_N}\circ T_{{\mathrm{SCI}}}  D_{0\ldots N}  (\nu \odot \omega_{k_1}^{(1)} \odot\ldots \odot \omega_{k_N}^{(N)})\nonumber \\
&= e_{j_0,\ldots,j_N}\circ T_{{\mathrm{SCI}}}  (\nu_0 \odot \omega_{k_1}^{(1)} \odot\ldots \odot \omega_{k_N}^{(N)})\nonumber \\
&= \sum_{i=0}^{d-1}\lambda_i e_{j_0,\ldots,j_N}\circ T_{{\mathrm{SCI}}}  (\omega_i^{(0)} \odot \omega_{k_1}^{(1)} \odot\ldots \odot \omega_{k_N}^{(N)})\nonumber \\
&= \sum_{i=0}^{d-1}\lambda_i e_{j_0,\ldots,j_N}(\omega_{i,i+k_1,\ldots,i+k_N})\nonumber \\
&=\lambda_{j_0} \delta_{j_1, j_0+k_1} \ldots \delta_{j_N, j_0 + k_N}.
\end{align}
Furthermore, $e_{j_0}^{(0)}(\nu)=e_{j_0}^{(0)}\circ D_0(\nu)={\lambda_{j_0}}$. This proves that $T_{{\mathrm{SCI}}} $ generates an {idealized} Darwinism process.

We summarize our findings in the following theorem:
\begin{theorem}
\label{TheDecImpliesD}
Consider a composition $\mathcal{S}\mathcal{E}_1\ldots\mathcal{E}_N$ of GPT systems $\mathcal{S},\mathcal{E}_1,\ldots,\mathcal{E}_N$ that carry $d$-outcome MCI-frames arising from decoherence. This composite system admits a transformation $T_{{\mathrm{SCI}}} \in\mathcal{T}_{\mathcal{S}\mathcal{E}_1\ldots\mathcal{E}_N}$ that robustly spreads classical information \emph{and} that generates an ideal Darwinism process.
\end{theorem}
Composite systems in quantum theory are unrestricted and satisfy strong symmetry (on states and effects). 
Furthermore, they admit MCI-frames arising from decoherence in the way specified above. 
Thus, the existence of an {idealized} Darwinism process and the robust spreading of classical information in quantum theory follow both as special cases of Theorem~\ref{TheSSImpliesD} and Theorem~\ref{TheDecImpliesD}. 
Apart from standard complex quantum theory, quantum theory over the real numbers is an example of a GPT that also abides by these requirements (and is not tomographically local).

\subsection{Darwinism in Spekkens' Toy Model}
\label{sec:STM}
If one identifies too many specific restrictions on a GPT, it raises the natural question: ``is quantum theory (or some subtheory thereof, such as real quantum theory) the only physical theory that allows for Darwinism?''
We answer this in the negative by providing an example that admits Darwinism, but is not quantum theory: Spekkens' Toy Model (STM)~\cite{Spekkens_ToyModel}. 

STM satisfies many of the same restrictions as quantum theory, such as no-signalling and no-cloning,
 and emulates many quantum behaviours such as complementary measurements, interference, entanglement (and monogamy thereof), and teleportation~\cite{Spekkens_ToyModel}.
Despite this, it is very different from quantum theory: both mathematically and conceptually, since at its core it is a classical {local} hidden-variable model.
What enables this quantum-like behaviour is that the states of maximum knowledge of the system are subject to the {\em epistemic restriction} that one knows only half of the possible information about the hidden {\em ontic} variable,
 along with a measurement-update rule that ensures that this restriction is maintained even when one makes sequential measurements on the system.

A more detailed description of STM and its extension into the GPT framework is given in \cref{app:Spekkens}.
For now, it suffices to remark that the composition of such systems is achieved by composing the underlying hidden classical variable (i.e.\ by Cartesian product) and applying the epistemic restriction to both the composite system and every subsystem thereof.

As observed by \citet{Pusey_Stabilizer} (and recounted in \cref{app:SpekkensStab}), the states within STM may be treated very similarly to the stabilizer subset of quantum theory (for a single system, the state spaces are isomorphic).
In particular, a single elementary STM system admits three ``toy observables'' $X$, $Y$ and $Z$ which act on the state to produce outputs $+1$ or $-1$ -- and there is one pure state for each of these six possibilities ($\ket{x\pm}, \ket{y\pm}, \ket{z\pm}$) and no other pure states.
When the ``wrong'' observable acts on a pure state (e.g.\ acting on $\ket{z+}$ with $X$), outcomes $+1$ and $-1$ occur with equal probability.
In this language, one can define the CNOT analogue for two STM bits ``control'' $C$ and ``target'' $T$:
 \begin{align}
	\mathrm{CNOT}: & X_C \mapsto X_C X_T, \ X_T \mapsto X_T, \nonumber \\
	& \quad Z_C \mapsto Z_C,\  Z_T \mapsto Z_C Z_T.
\end{align}
This can be read as, e.g.\ $X_C\mapsto X_C X_T$, ``The product of the observation of $X$ on $C$ and $X$ on $T$ after the transformation $\mathrm{CNOT}$ yields the same outcome statistics as the observation $X$ on $C$ before the transformation.''

With this shorthand, we hence specify our candidate for an ideal Darwinism process from main system $S$ onto multiple environments $E_1,...,E_N$:
\begin{align}
\label{eq:stabCNOTmany}
	\mathrm{FAN}: & \quad X_S \mapsto X_S X_{E_1}...X_{E_n}, \quad \forall k:  X_{E_k} \mapsto X_{E_k}, \nonumber \\
	& \quad Z_S \mapsto Z_S, \quad \forall k: Z_{E_k} \mapsto Z_S Z_{E_{k}}.
\end{align}
The validity of this, as a transformation in STM, can be verified in one of two ways:
 the first is to consider a direct implementation of this as a series of pairwise CNOT gates (in the manner of \cref{fig:fanout}),
 reasoning (e.g.\ via category theory~\cite{Coecke_STM}) that such composition is permissible.
The second way is to note that this map is admissible as a transformation on an analogously defined $N$-bit quantum stabilizer system,
 and then use the result of \citet{Pusey_Stabilizer} to infer that this makes ${\rm FAN}$ a valid STM transformation.

Thus, it remains to verify that such a transformation indeed achieves the desired ideal Darwinistic behaviour and robustly spreads classical information.
Suppose we have an initial state of the form $\ket{\psi}_S\otimes\ket{{z+}}^{\otimes n}_{E_1 ... E_N}$ where $\ket{\psi}_S$ is some arbitrary pure STM bit state of the main system, and $\ket{z+}$ corresponds to the state that always gives output $+1$ when measured by toy observable $Z$.
As for each $k$, FAN maps $Z_{E_k}$ to $Z_S Z_{E_k}$, the final state will always have result $+1$ for joint measurements of $Z_S Z_{E_k}$ -- {mandating that} the results of $Z_S$ and $Z_{E_k}$ are perfectly correlated. (In the case $E_j$ starts at $\ket{z-}$, anti-correlation is established.)
Therefore the fan-out results in all observers seeing the same outcome as made on the original system. 

Our other requirement for Darwinism is that the outcome probability of $Z_S$ is not changed, and this is also explicitly given by the rule in the map $Z_S\mapsto Z_S$.
In particular $\ket{z\pm}$ are the only pure states that have non-zero expectation value for the observable $Z$,
 and the map does not take any state of main system stabilized by another observable (i.e. $X$ or $Y$) to any state stabilized by an expression containing $Z$.
As such, since $S$ can only be in one of these possibilities (or convex combination thereof in the GPT extension) this implies that the statistics of $Z_\mathcal{S}$ remain unchanged. 

Finally, this transformation also enables robust spreading of classical information. Indeed, allow some of the environmental systems to be initialized in $\ket{z-}_{E_k}$ instead of $\ket{z+}_{E_k}$. Then, the conditions $Z_S\mapsto Z_S$ and $(-) Z_{E_k}\mapsto (-) Z_SZ_{E_k}$  imply that, if the system is initialized in $\ket{\psi}_S=\ket{z+}$ no change occurs in the total state;  if $\ket{\psi}_S=\ket{z-}$, then the state of the system is also kept unchanged but all environmental subsystems flip sign.
We summarize this with our final theorem of the article:
\begin{theorem}[STM admits an ideal Darwinism process]
\label{thm:STM}
The FAN operation specified in \cref{eq:stabCNOTmany} implements an {idealized} Darwinism process, as per \cref{def:DarwinGPT}, and robustly spreads classical information (definition \ref{def:SpreadingCI}).
\end{theorem}

We conclude this section with some remarks on the implications of this example to the theorems of this paper.
First, STM is nonclassical in that it violates \cref{def:qcMCI}: e.g.,\ for the MCI--frame distinguishing measurement $Z$, the pure states $\ket{x+}$, $\ket{x-}$, $\ket{y+}$ and $\ket{y-}$ do not give a deterministic response. Thus, as per Theorem~\ref{theorem:DneedsEntangledStates}, the FAN operation creates entangled states from separable ones. E.g.,\ for a system and single environmental bit, the FAN (i.e.\ CNOT) gate takes the separable state $\ket{x+}_S\ket{z+}_E$ to a ``maximally entangled'' state, where the ontic state of the two systems is guaranteed to be the same (i.e.\ the same outcomes will be observed if the same measurement is made on both systems).
STM is also nonclassical in the sense of \cref{def:inequivalent}, in that there are more than one set of sufficiently different refined measurements, and indeed also STM has entangled effects as mandated by Theorem~\ref{thm:EntangledEffects}.
Although this demonstrates the necessity of entangled effects in a non-classical setting, we can further conclude (by counterexample) that the stronger condition of violating Bell inequalities (see e.g.~\cite{Brunner_Bell}) is {\em not} necessary since STM does not violate these. 
A similar conclusion follows for \emph{contextuality}, which is not present in STM \cite{Spekkens_ToyModel} and thus shown to be unnecessary for Darwinism.

Secondly, in terms of the sufficient conditions, STM neither admits a decoherence map, nor is it strongly symmetric (as we show in \cref{app:SpekkSuff}).
This illustrates that the sufficient conditions are not tight -- they enable the fan-out dynamic by mandating the existence of {\em all} classical dynamics within the theory.
However, the fan-out operation can be admitted without requiring universal classical computation -- indeed, as above for STM, or existing as a member of the (non-universal~\cite{Gottesman_stab}) Clifford group in the case of quantum stabilizers.

\section{Conclusions and perspectives}
Quantum Darwinism provides a mechanism through which crucial aspects of classicality can be understood to emerge in the quantum domain~\cite{Zurek_QuantumOrigins,Zurek_QDReview1,Zurek_QDReview2,Brandao_QD,Knott_QDInfiniteDimension}. 
In this article, we generalized an {idealized} notion of Darwinism, where maximal classical information is perfectly broadcast to an environment split into fractions, to the framework of GPTs. 
We showed that entanglement, in both states and measurements, is a necessary feature for such a process to be present in generalized theories, and demonstrate{d} that some important physical principles -- like strong symmetry and decoherence -- provide sufficient structure to admit Darwinism. 
Finally, we described {a mechanism for} Darwinism in Spekkens Toy Model, showing that {such broadcasting of classical information is not unique to} quantum theory{, and, moreover, that our sufficient conditions are not tight}. 

Our results show that objectivity may arise through a Darwinism process in non-classical theories other than quantum {theory} -- adding to the results of \citet{Scandolo_ObjectivityGPT}, which analyzed objectivity through State Spectrum Broadcast in GPTs. 
Complementing a previous result on decoherence~\cite{Richens_DecoherenceGPTs}, our work also shows the important role of entanglement to allow for emergence of classicality, suggesting the counterintuitive principle that locally non-classical theories must also allow for shared non-classicality to allow for the emergence of classical objectivity.
In addition, our results show that strongly symmetric and unrestricted GPTs -- that is, those endowed with sufficient structure to allow for reversible classical computation and the encoding and decoding of classical information -- have sufficient structure for Darwinism to be present.

An important outlook is whether additional features often associated with Quantum Darwinism can also be generalized into this framework, or whether our results also hold for Darwinist processes with less structure than the idealized fan-out transformations considered here.
For instance, can we prove similar results if one requires only one initial state of the environment to be a good register (i.e., for minimal Darwinism)? 
We know this to be the case for Theorem~\ref{theorem:DneedsEntangledStates} (see the remark in this theorem), but this is not yet shown for the other results. 
Additionally, in the spirit of ref.~\cite{Zwolak_NoisyChannels}, what happens if the initial state of the environment is not able to perfectly encode the central system's information? 
One could also maintain such idealized features but look into the case in which the classical information that spreads to the environment is non-maximal, as a GPT analog to quantum dynamics that lead to decoherence-free subspaces.

Finally, although this work has been presented with a focus on the origins of classical limits, our results also have a bearing on the general foundations of computation~\cite{Lee_GPTComp,Garner_IXComp}.
The Darwinism--enabling fan-out transformation (\cref{eq:QPointerStates}) has its origins in classical logic circuits, connecting the output of one logic gate to the input of many others, and its quantum analogue plays a role in the design of quantum neural networks~\cite{Wan_NN}. The conclusions of this article therefore imply that such computation also necessitates the existence of entanglement, if the theory is not strictly classical 
 -- meanwhile identifying potential sufficient structure (e.g.\ no-restriction and strong symmetry) to guarantee that such computation can be performed.

\section*{Acknowledgments}
RDB acknowledges funding by S\~{a}o Paulo Research Foundation -- FAPESP, through scholarships no.\ 2016/24162-8 and no.\ 2019/02221-0. RDB is also thankful to Marcelo Terra Cunha for insightful ideas and discussions, and to IQOQI for the hospitality during the time as guest researcher. MK acknowledges the support of the Vienna Doctoral School in Physics (VDSP) and the support of the Austrian Science Fund (FWF) through the Doctoral Programme CoQuS. MK, AJPG and MPM thank the Foundational Questions Institute and Fetzer Franklin Fund, a donor advised fund of Silicon Valley Community Foundation, for support via grant number FQXi-RFP-1815. This research was supported in part by Perimeter Institute for Theoretical Physics. Research at Perimeter Institute is supported by the Government of Canada through the Department of Innovation, Science and Economic Development Canada and by the Province of Ontario through the Ministry of Research, Innovation and Science.

\bibliographystyle{unsrtnat}
\bibliography{gptd}

\begin{thebibliography}{59}
\providecommand{\natexlab}[1]{#1}
\providecommand{\url}[1]{\texttt{#1}}
\expandafter\ifx\csname urlstyle\endcsname\relax
  \providecommand{\doi}[1]{doi: #1}\else
  \providecommand{\doi}{doi: \begingroup \urlstyle{rm}\Url}\fi

\bibitem[Zurek(2003)]{Zurek_QuantumOrigins}
W.~H. Zurek.
\newblock Decoherence, einselection, and the quantum origins of the classical.
\newblock \emph{Rev. Mod. Phys.}, 75\penalty0 (3):\penalty0 715--775, 2003.
\newblock \doi{10.1103/revmodphys.75.715}.

\bibitem[Zurek(2007)]{Zurek_QDReview1}
W.~H. Zurek.
\newblock {Relative States and the Environment: Einselection, Envariance,
  Quantum Darwinism, and the Existential Interpretation}.
\newblock \emph{Pre-print}, arXiv:0707.2832, 2007.
\newblock URL \url{{https://arxiv.org/abs/0707.2832}}.

\bibitem[Zurek(2009)]{Zurek_QDReview2}
W.~H. Zurek.
\newblock {Quantum Darwinism}.
\newblock \emph{Nature Physics}, 5\penalty0 (3):\penalty0 181--188, 2009.
\newblock \doi{10.1038/nphys1202}.

\bibitem[Brand{\~a}o et~al.(2015)Brand{\~a}o, Piani, and Horodecki]{Brandao_QD}
F.~G. S.~L. Brand{\~a}o, M.~Piani, and P.~Horodecki.
\newblock Generic emergence of classical features in quantum darwinism.
\newblock \emph{Nature Communications}, 6\penalty0 (1), 2015.
\newblock \doi{10.1038/ncomms8908}.

\bibitem[Knott et~al.(2018)Knott, Tufarelli, Piani, and
  Adesso]{Knott_QDInfiniteDimension}
P.~A. Knott, T.~Tufarelli, M.~Piani, and G.~Adesso.
\newblock Generic emergence of objectivity of observables in infinite
  dimensions.
\newblock \emph{Phys. Rev. Lett.}, 121:\penalty0 160401, 2018.
\newblock \doi{10.1103/PhysRevLett.121.160401}.

\bibitem[Horodecki et~al.(2015)Horodecki, Korbicz, and
  Horodecki]{Horodecki_NoMI}
R.~Horodecki, J.~K. Korbicz, and P.~Horodecki.
\newblock {Quantum origins of objectivity}.
\newblock \emph{Phys. Rev. A}, 91\penalty0 (3):\penalty0 032122, 2015.
\newblock \doi{10.1103/PhysRevA.91.032122}.

\bibitem[Le and Olaya-Castro(2018)]{Le_ObjectivityQDxSSB}
T.~P. Le and A.~Olaya-Castro.
\newblock Objectivity (or lack thereof): Comparison between predictions of
  quantum darwinism and spectrum broadcast structure.
\newblock \emph{Phys. Rev. A}, 98\penalty0 (3), 2018.
\newblock \doi{10.1103/physreva.98.032103}.

\bibitem[Blume-Kohout and Zurek(2008)]{Blume-Kohout_QDBrownian}
R.~Blume-Kohout and W.~H. Zurek.
\newblock Quantum darwinism in quantum brownian motion.
\newblock \emph{Phys. Rev. Lett.}, 101:\penalty0 240405, 2008.
\newblock \doi{10.1103/PhysRevLett.101.240405}.

\bibitem[Blume-Kohout and Zurek(2005)]{Blume-Kohout_RedundantInfo}
R.~Blume-Kohout and W.~H. Zurek.
\newblock A simple example of ``quantum darwinism'': Redundant information
  storage in many-spin environments.
\newblock \emph{Foundations of Physics}, 35\penalty0 (11):\penalty0 1857--1876,
  2005.
\newblock \doi{10.1007/s10701-005-7352-5}.

\bibitem[Riedel and Zurek(2010)]{Riedel_QDIllumination}
C.~J. Riedel and W.~H. Zurek.
\newblock Quantum darwinism in an everyday environment: Huge redundancy in
  scattered photons.
\newblock \emph{Phys. Rev. Lett.}, 105:\penalty0 020404, 2010.
\newblock \doi{10.1103/PhysRevLett.105.020404}.

\bibitem[Zwolak et~al.(2010)Zwolak, Quan, and Zurek]{Zwolak_NoisyChannels}
M.~Zwolak, H.~T. Quan, and W.~H. Zurek.
\newblock Redundant imprinting of information in nonideal environments:
  Objective reality via a noisy channel.
\newblock \emph{Phys. Rev. A}, 81:\penalty0 062110, 2010.
\newblock \doi{10.1103/PhysRevA.81.062110}.

\bibitem[Zwolak et~al.(2009)Zwolak, Quan, and Zurek]{Zwolak_MixedEnvironment}
M.~Zwolak, H.~T. Quan, and W.~H. Zurek.
\newblock Quantum darwinism in a mixed environment.
\newblock \emph{Phys. Rev. Lett.}, 103:\penalty0 110402, 2009.
\newblock \doi{10.1103/PhysRevLett.103.110402}.

\bibitem[Dirac(1958)]{DiracBook}
P.~A.~M. Dirac.
\newblock \emph{The Principles of Quantum Mechanics}.
\newblock Oxford University Press, Oxford, 4th edition, 1958.

\bibitem[Sakurai and Napolitano(2011)]{Sakurai}
J.~J. Sakurai and J.~Napolitano.
\newblock \emph{Modern Quantum Mechanics}.
\newblock Addison-Wesley, 2nd edition, 2011.

\bibitem[Schlosshauer(2004)]{Schlosshauer_Deco}
M.~Schlosshauer.
\newblock Decoherence, the measurement problem, and interpretations of quantum
  mechanics.
\newblock \emph{Rev. Mod. Phys.}, 76\penalty0 (4):\penalty0 1267--1305, 2004.
\newblock \doi{10.1103/RevModPhys.76.1267}.

\bibitem[Wootters and Zurek(1982)]{Wootters_NoCloning}
W.~K. Wootters and W.~H. Zurek.
\newblock A single quantum cannot be cloned.
\newblock \emph{Nature}, 299\penalty0 (5886):\penalty0 802--803, 1982.
\newblock \doi{10.1038/299802a0}.

\bibitem[Hardy(2001)]{Hardy_Axioms}
L.~Hardy.
\newblock {Quantum Theory From Five Reasonable Axioms}.
\newblock \emph{{Pre-print}}, {arXiv:quant-ph/0101012}, 2001.
\newblock URL \url{{https://arxiv.org/abs/quant-ph/0101012}}.

\bibitem[Barrett(2007)]{Barrett_GPT}
J.~Barrett.
\newblock Information processing in generalized probabilistic theories.
\newblock \emph{Phys. Rev. A}, 75:\penalty0 032304, 2007.
\newblock \doi{10.1103/PhysRevA.75.032304}.

\bibitem[Barnum et~al.(2012)Barnum, Barrett, Leifer, and
  Wilce]{Barnum_Teleport}
H.~Barnum, J.~Barrett, M.~Leifer, and A.~Wilce.
\newblock Teleportation in general probabilistic theories, 2012.

\bibitem[Barnum et~al.(2007)Barnum, Barrett, Leifer, and
  Wilce]{Barnum_NoBroadcasting}
H.~Barnum, J.~Barrett, M.~Leifer, and A.~Wilce.
\newblock Generalized no-broadcasting theorem.
\newblock \emph{Phys. Rev. Lett.}, 99\penalty0 (24), 2007.
\newblock \doi{10.1103/physrevlett.99.240501}.

\bibitem[Garner et~al.(2013)Garner, Dahlsten, Nakata, Murao, and
  Vedral]{Garner_GPTPhase}
A.~J.~P. Garner, O.~C.~O. Dahlsten, Y.~Nakata, M.~Murao, and V.~Vedral.
\newblock {A framework for phase and interference in generalized probabilistic
  theories}.
\newblock \emph{New J. Phys.}, 15\penalty0 (9):\penalty0 093044, 2013.
\newblock \doi{10.1088/1367-2630/15/9/093044}.

\bibitem[Dahlsten et~al.(2014)Dahlsten, Garner, and Vedral]{Dahlsten_BL}
O.~C.~O. Dahlsten, A.~J.~P. Garner, and V.~Vedral.
\newblock {The uncertainty principle enables non-classical dynamics in an
  interferometer}.
\newblock \emph{Nature Communications}, 5\penalty0 (4592), 2014.
\newblock \doi{10.1038/ncomms5592}.

\bibitem[Richens et~al.(2017)Richens, Selby, and
  Al-Safi]{Richens_DecoherenceGPTs}
J.~G. Richens, J.~H. Selby, and S.~W. Al-Safi.
\newblock Entanglement is necessary for emergent classicality in all physical
  theories.
\newblock \emph{Phys. Rev. Lett.}, 119:\penalty0 080503, 2017.
\newblock \doi{10.1103/PhysRevLett.119.080503}.

\bibitem[Chiribella et~al.(2010)Chiribella, D'Ariano, and
  Perinotti]{Chiribella_Purification}
G.~Chiribella, G.~M. D'Ariano, and P.~Perinotti.
\newblock Probabilistic theories with purification.
\newblock \emph{Phys. Rev. A}, 81\penalty0 (6), 2010.
\newblock \doi{10.1103/physreva.81.062348}.

\bibitem[Janotta and Lal(2013)]{Janotta_GPTsRestriction}
P.~Janotta and R.~Lal.
\newblock Generalized probabilistic theories without the no-restriction
  hypothesis.
\newblock \emph{Phys. Rev. A}, 87\penalty0 (5), 2013.
\newblock \doi{10.1103/physreva.87.052131}.

\bibitem[Barnum et~al.(2014{\natexlab{a}})Barnum, M{\"{u}}ller, and
  Ududec]{BMU_Ix}
H.~Barnum, M.~P. M{\"{u}}ller, and C.~Ududec.
\newblock {Higher-order interference and single-system postulates
  characterizing quantum theory}.
\newblock \emph{New J. Phys.}, 16\penalty0 (12):\penalty0 123029,
  2014{\natexlab{a}}.
\newblock \doi{10.1088/1367-2630/16/12/123029}.

\bibitem[Spekkens(2007)]{Spekkens_ToyModel}
R.~W. Spekkens.
\newblock Evidence for the epistemic view of quantum states: A toy theory.
\newblock \emph{Phys. Rev. A}, 75\penalty0 (3):\penalty0 32110, 2007.
\newblock \doi{10.1103/PhysRevA.75.032110}.

\bibitem[Holevo(1973)]{Holevo_Bound}
A.~S. Holevo.
\newblock {Bounds for the quantity of information transmitted by a quantum
  communication channel}.
\newblock \emph{Probl. Peredachi Inf.}, 9\penalty0 (3):\penalty0 177--183,
  1973.

\bibitem[Zwolak et~al.(2016)Zwolak, Riedel, and Zurek]{Zwolak_Amplification}
M.~Zwolak, C.~J. Riedel, and W.~H. Zurek.
\newblock Amplification, decoherence and the acquisition of information by spin
  environments.
\newblock \emph{Scientific Reports}, 6\penalty0 (1):\penalty0 25277, 2016.
\newblock \doi{10.1038/srep25277}.

\bibitem[Unden et~al.(2019)Unden, Louzon, Zwolak, Zurek, and
  Jelezko]{Unden_QDNV}
T.~K. Unden, D.~Louzon, M.~Zwolak, W.~H. Zurek, and F.~Jelezko.
\newblock Revealing the emergence of classicality using nitrogen-vacancy
  centers.
\newblock \emph{Phys. Rev. Lett.}, 123:\penalty0 140402, 2019.
\newblock \doi{10.1103/PhysRevLett.123.140402}.

\bibitem[Ciampini et~al.(2018)Ciampini, Pinna, Mataloni, and
  Paternostro]{Ciampini_ExperimentalQD}
M.~A. Ciampini, G.~Pinna, P.~Mataloni, and M.~Paternostro.
\newblock Experimental signature of quantum darwinism in photonic cluster
  states.
\newblock \emph{Phys. Rev. A}, 98:\penalty0 020101, 2018.
\newblock \doi{10.1103/PhysRevA.98.020101}.

\bibitem[Riedel and Zurek(2011)]{Riedel_QDIlluminationExtended}
C.~J. Riedel and W.~H. Zurek.
\newblock Redundant information from thermal illumination: quantum darwinism in
  scattered photons.
\newblock \emph{New Journal of Physics}, 13\penalty0 (7):\penalty0 073038, jul
  2011.
\newblock \doi{10.1088/1367-2630/13/7/073038}.

\bibitem[Mazurek et~al.(2021)Mazurek, Pusey, Resch, and
  Spekkens]{MazurekPRS_Reconstruction}
M.~D. Mazurek, M.~F. Pusey, K.~J. Resch, and R.~W. Spekkens.
\newblock Experimentally bounding deviations from quantum theory in the
  landscape of generalized probabilistic theories.
\newblock \emph{PRX Quantum}, 2\penalty0 (2), Apr 2021.
\newblock ISSN 2691-3399.
\newblock \doi{10.1103/prxquantum.2.020302}.

\bibitem[Chiribella et~al.(2011)Chiribella, D'Ariano, and
  Perinotti]{Chiribella_Axioms}
G.~Chiribella, G.~M. D'Ariano, and P.~Perinotti.
\newblock {Informational derivation of quantum theory}.
\newblock \emph{Phys. Rev. A}, 84\penalty0 (1):\penalty0 012311, 2011.
\newblock \doi{10.1103/PhysRevA.84.012311}.

\bibitem[Masanes and M{\"u}ller(2011)]{Masanes_DerivQT}
Ll. Masanes and M.~P. M{\"u}ller.
\newblock A derivation of quantum theory from physical requirements.
\newblock \emph{New J. Phys.}, 13\penalty0 (6):\penalty0 063001, 2011.
\newblock \doi{10.1088/1367-2630/13/6/063001}.

\bibitem[M{\"{u}}ller and Masanes(2013)]{Mueller_3DSpace}
M.~P. M{\"{u}}ller and Ll. Masanes.
\newblock {Three-dimensionality of space and the quantum bit: an
  information-theoretic approach}.
\newblock \emph{New J. Phys.}, 15\penalty0 (5):\penalty0 053040, 2013.
\newblock \doi{10.1088/1367-2630/15/5/053040}.

\bibitem[Müller(2021)]{Mueller_LesHouches}
M.~P. Müller.
\newblock {Probabilistic Theories and Reconstructions of Quantum Theory (Les
  Houches 2019 lecture notes)}.
\newblock \emph{SciPost Phys. Lect. Notes}, page~28, 2021.
\newblock \doi{10.21468/SciPostPhysLectNotes.28}.

\bibitem[Webster(1994)]{Webster94}
R.~Webster.
\newblock \emph{{Convexity}}.
\newblock Oxford University Press, Oxford, 1994.
\newblock ISBN 0-19-853147-8.

\bibitem[Chiribella and Scandolo(2015)]{Chiribella_PP}
G.~Chiribella and C.~M. Scandolo.
\newblock {Operational axioms for diagonalizing states}.
\newblock \emph{Electronic Proceedings in Theoretical Computer Science},
  195:\penalty0 96--115, 2015.
\newblock \doi{10.4204/EPTCS.195.8}.

\bibitem[M\"uller and Ududec(2012)]{Mueller_Bit-Symmetry}
M.~P. M\"uller and C.~Ududec.
\newblock Structure of reversible computation determines the self-duality of
  quantum theory.
\newblock \emph{Phys. Rev. Lett.}, 108:\penalty0 130401, 2012.
\newblock \doi{10.1103/PhysRevLett.108.130401}.

\bibitem[Al-Safi and Richens(2015)]{Al_Safi_ReversibleDichotomyOddN}
S.~W. Al-Safi and J.~Richens.
\newblock Reversibility and the structure of the local state space.
\newblock \emph{New J. Phys.}, 17\penalty0 (12):\penalty0 123001, 2015.
\newblock \doi{10.1088/1367-2630/17/12/123001}.

\bibitem[Massar and Patra(2014)]{Massar_InfoPolygon}
S.~Massar and M.~K. Patra.
\newblock Information and communication in polygon theories.
\newblock \emph{Phys. Rev. A}, 89:\penalty0 052124, 2014.
\newblock \doi{10.1103/PhysRevA.89.052124}.

\bibitem[Janotta and Hinrichsen(2014)]{Janotta_GPT}
P.~Janotta and H.~Hinrichsen.
\newblock Generalized probability theories: what determines the structure of
  quantum theory?
\newblock \emph{Journal of Physics A: Mathematical and Theoretical},
  47\penalty0 (32):\penalty0 323001, 2014.
\newblock \doi{10.1088/1751-8113/47/32/323001}.

\bibitem[Coecke and Heunen(2016)]{Coecke_CP}
B.~Coecke and C.~Heunen.
\newblock {Pictures of complete positivity in arbitrary dimension}.
\newblock \emph{Information and Computation}, 250:\penalty0 50--58, 2016.
\newblock \doi{10.1016/j.ic.2016.02.007}.

\bibitem[Gross et~al.(2010)Gross, M\"uller, Colbeck, and
  Dahlsten]{Gross_TBoxworld1}
D.~Gross, M.~P. M\"uller, R.~Colbeck, and O.~C.~O. Dahlsten.
\newblock All reversible dynamics in maximally nonlocal theories are trivial.
\newblock \emph{Phys. Rev. Lett.}, 104:\penalty0 080402, 2010.
\newblock \doi{10.1103/PhysRevLett.104.080402}.

\bibitem[Barnum et~al.(2014{\natexlab{b}})Barnum, M{\"u}ller, and
  Ududec]{Barnum_HigherOrder}
H.~Barnum, M.~P. M{\"u}ller, and C.~Ududec.
\newblock Higher-order interference and single-system postulates characterizing
  quantum theory.
\newblock \emph{New J. Phys.}, 16\penalty0 (12):\penalty0 123029,
  2014{\natexlab{b}}.
\newblock \doi{10.1088/1367-2630/16/12/123029}.

\bibitem[Pusey(2012)]{Pusey_Stabilizer}
M.~F. Pusey.
\newblock {Stabilizer Notation for Spekkens' Toy Theory}.
\newblock \emph{Foundations of Physics}, 42\penalty0 (5):\penalty0 688--708,
  2012.
\newblock \doi{10.1007/s10701-012-9639-7}.

\bibitem[Coecke et~al.(2011)Coecke, Edwards, and Spekkens]{Coecke_STM}
B.~Coecke, B.~Edwards, and R.~W. Spekkens.
\newblock {Phase Groups and the Origin of Non-locality for Qubits}.
\newblock \emph{Electronic Notes in Theoretical Computer Science}, 270\penalty0
  (2):\penalty0 15--36, 2011.
\newblock \doi{10.1016/j.entcs.2011.01.021}.

\bibitem[Brunner et~al.(2014)Brunner, Cavalcanti, Pironio, Scarani, and
  Wehner]{Brunner_Bell}
N.~Brunner, D.~Cavalcanti, S.~Pironio, V.~Scarani, and S.~Wehner.
\newblock {Bell nonlocality}.
\newblock \emph{Reviews of Modern Physics}, 86\penalty0 (2):\penalty0 419--478,
  2014.
\newblock \doi{10.1103/RevModPhys.86.419}.

\bibitem[Gottesman(1999)]{Gottesman_stab}
D.~Gottesman.
\newblock {The Heisenberg Representation of Quantum Computers}.
\newblock In \emph{Proceedings of the XXII International Colloquium on Group
  Theoretical Methods in Physics}, pages 32--43, 1999.
\newblock URL \url{https://arxiv.org/abs/quant-ph/9807006}.

\bibitem[Scandolo et~al.(2021)Scandolo, Salazar, Korbicz, and
  Horodecki]{Scandolo_ObjectivityGPT}
C.~M. Scandolo, R.~Salazar, J.~K. Korbicz, and P.~Horodecki.
\newblock Universal structure of objective states in all fundamental causal
  theories.
\newblock \emph{Physical Review Research}, 3\penalty0 (3), Aug 2021.
\newblock ISSN 2643-1564.
\newblock \doi{10.1103/physrevresearch.3.033148}.

\bibitem[Lee and Barrett(2015)]{Lee_GPTComp}
C.~M. Lee and J.~Barrett.
\newblock {Computation in generalised probabilisitic theories}.
\newblock \emph{New J. Phys.}, 17\penalty0 (17), 2015.
\newblock \doi{10.1088/1367-2630/17/8/083001}.

\bibitem[Garner(2018)]{Garner_IXComp}
A.~J.~P. Garner.
\newblock {Interferometric Computation Beyond Quantum Theory}.
\newblock \emph{Foundations of Physics}, 2018.
\newblock \doi{10.1007/s10701-018-0142-7}.

\bibitem[Wan et~al.(2017)Wan, Dahlsten, Kristj{\'{a}}nsson, Gardner, and
  Kim]{Wan_NN}
K.~H. Wan, O.~C.~O. Dahlsten, H.~Kristj{\'{a}}nsson, R.~Gardner, and M.~S. Kim.
\newblock {Quantum generalisation of feedforward neural networks}.
\newblock \emph{npj Quantum Information}, 3\penalty0 (1), 2017.
\newblock \doi{10.1038/s41534-017-0032-4}.

\bibitem[Janotta et~al.(2011)Janotta, Gogolin, Barrett, and
  Brunner]{Janotta_PolygonSpaces}
P.~Janotta, C.~Gogolin, J.~Barrett, and N.~Brunner.
\newblock Limits on nonlocal correlations from the structure of the local state
  space.
\newblock \emph{New J. Phys.}, 13\penalty0 (6):\penalty0 063024, 2011.
\newblock \doi{10.1088/1367-2630/13/6/063024}.

\bibitem[Masanes et~al.(2013)Masanes, M{\"{u}}ller, Augusiak, and
  P{\'{e}}rez-Garc{\'{i}}a]{Masanes_InfoUnit}
Ll. Masanes, M.~P. M{\"{u}}ller, R.~Augusiak, and D.~P{\'{e}}rez-Garc{\'{i}}a.
\newblock Existence of an information unit as a postulate of quantum theory.
\newblock \emph{Proceedings of the National Academy of Sciences of the United
  States of America}, 110\penalty0 (41):\penalty0 16373--16377, 2013.
\newblock \doi{10.1073/pnas.1304884110}.

\bibitem[Hardy(1999)]{hardy_ConvexSpekkensTT}
L.~Hardy.
\newblock Disentangling nonlocality and teleportation.
\newblock \emph{Pre-print}, arXiv:quant-ph/9906123, 1999.
\newblock URL \url{https://arxiv.org/abs/quant-ph/9906123}.

\bibitem[Garner(2015)]{Garner_DPhil}
A.~J.~P. Garner.
\newblock \emph{{Phase and interference phenomena in generalised probabilistic
  theories}}.
\newblock PhD thesis, University of Oxford, 2015.
\newblock URL
  \url{{https://ora.ox.ac.uk/objects/uuid:c0017faf-cbe0-4365-a1ff-080fa031d006}}.

\bibitem[Nielsen and Chuang(2000)]{Nielsen_Chuang}
M.~A. Nielsen and I.~L. Chuang.
\newblock \emph{{Quantum Computation and Quantum Information}}.
\newblock Cambridge University Press, Cambridge, 2000.
\newblock ISBN 0521635039.

\end{thebibliography}

\bigskip

\newpage
\appendix
\section*{Appendix}
\section{The pentagon state space}
\label{app:Examples}

We present an example of a state space~\cite{Massar_InfoPolygon} (brought to our attention in \citet{Janotta_PolygonSpaces}) without an MCI-frame (\cref{def:MCIframe}), and illustrate its counterintuitive properties.
\begin{example}[Pentagon state space]
\label{ex:Pentagon}
Consider a GPT system with states in $A=\mathbb{R}^3$ such that $\Omega_A$ is a regular pentagon (with pure states being the vertices), and with a dual space of effects $E_A$ subject to the no-restriction hypothesis.
Such a system admits a {\em self-dual} identification between $A_+$ and $E_A$ in the following sense: for each vertex $\omega_j$, there is a unique related effect $e_j\in E_A$ with $e_j\leq u_A$ such that $e_j(\nu)=1 \Rightarrow \nu=\omega_j$; 
 that is, these effects are in one-to-one correspondence with the vertices -- and those are exactly the {pure} effects. 
\end{example}

\begin{figure}[tbh]
    \centering
    \includegraphics[width=0.65\linewidth]{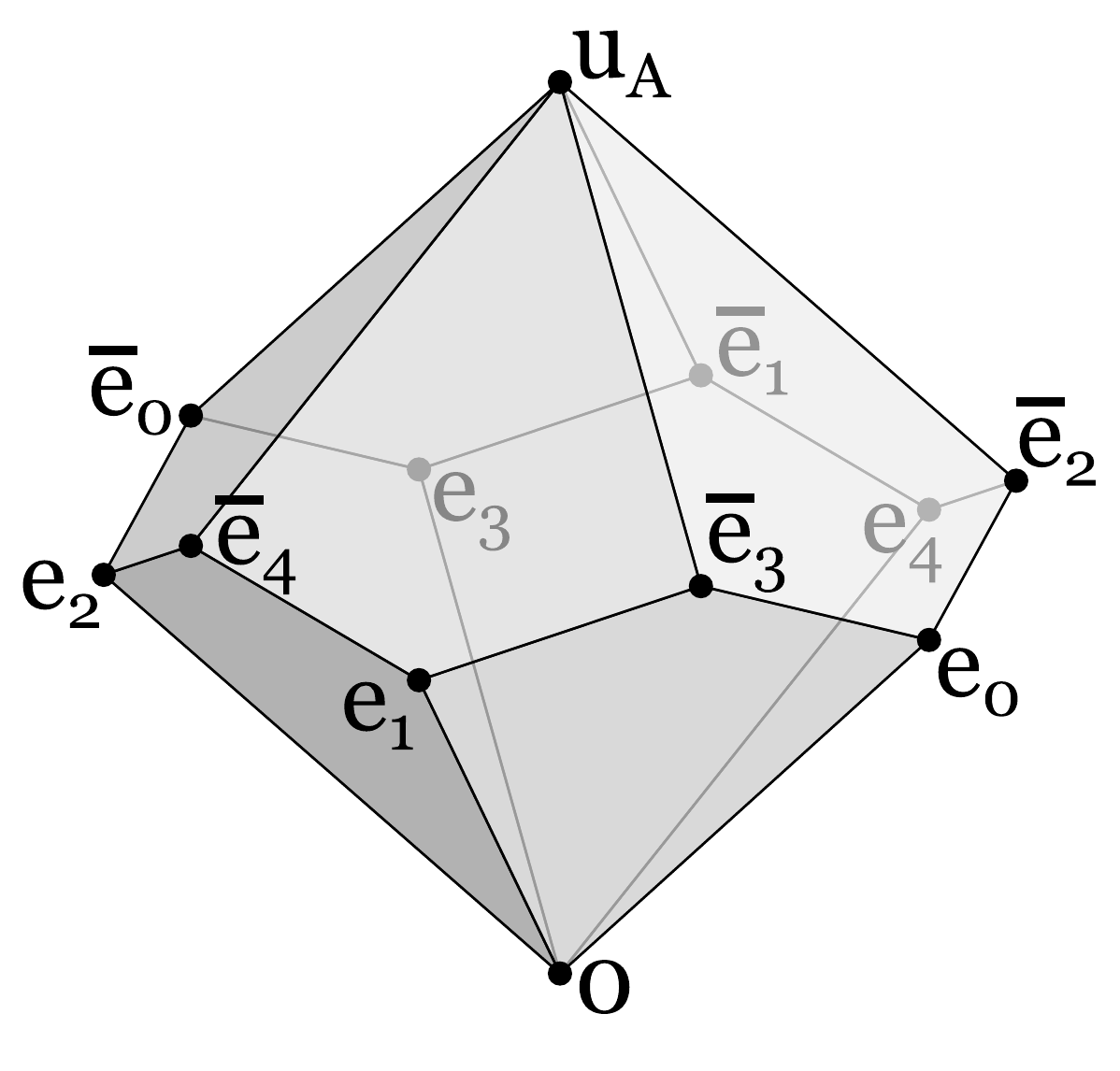}
    \caption{{
    \caphead{The pentagon effect space.}
    The pure effects (sitting on the extremal rays of the cone) are labeled $\{e_i\}$, and the unit effect is $u_A$. The complementary effects $\{\bar{e}_i := u_A - e_i\}$ are not pure, nor are they are convex combinations of $\{e_i\}$ -- rather, they are linear combinations.}}
    \label{fig:pentagon}
\end{figure}

\begin{figure}[tbh]
    \centering
    \includegraphics[width=0.625\linewidth]{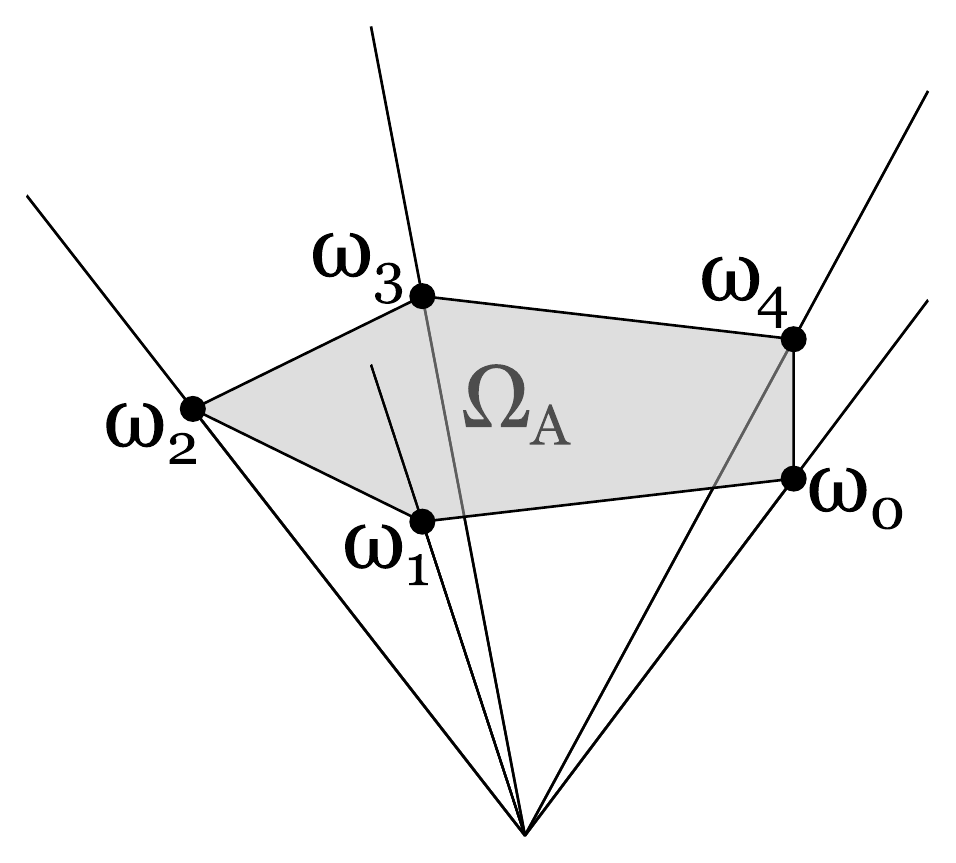}
    \caption{
    \caphead{The pentagon state space.}
    Normalized states $\Omega_A$ are shaded gray.
    The pure states (labeled $\{\omega_i\}$) correspond to the vertices of a regular pentagon.}
    \label{fig:pentagonS}
\end{figure}

Consider the effect space of this theory {(drawn in figure \ref{fig:pentagon})} where we have labeled the {pure effects} clockwise. 
Consider also the set of states (drawn in figure \ref{fig:pentagonS}), labeled in a similar way.
The maximal frame is of size two: any pair of {pure state} {whose absolute difference between indices is $2$ (modulo $5$) lie on ``opposite'' sides of the pentagon, and} form such a frame. 
Then, both {$\{\omega_0,\omega_2\}$} and {$\{\omega_0,\omega_3\}$ } are maximal frames; their states are distinguished, for example, by $M_1=\{e_0,\bar{e}_0\}$. 
However $\bar{e}_0$ is {\em not {pure}}: {it is possible to show (calculating explicitly using the representation described in \citet{Janotta_PolygonSpaces}) that} $\bar{e}_0= \alpha(e_2+e_3)$ with $\alpha{ = \frac{{\rm sec}(\pi/5)}{2}}\approx 0.6180 > 1/2$. {Thus, one can} also perform the refined measurement $M_2=\{e_0,\alpha e_2,\alpha e_3\}$ to distinguish the states in either frame $\{\omega_0, \omega_2\}$ or $\{\omega_0, \omega_3\}$.
{Here, if $\omega_0$ is prepared and measurement $M_2$ made, the outcome associated with $e_0$ will always be measured.
However when $\omega_2$ is prepared and $M_2$ measured, with probability $\alpha$ one will get the outcome associated with $e_2$, 
 and with probability $1-\alpha$ the (``incorrect'') outcome associated with $e_3$ (but one never gets the outcome associated with $e_0$).
(Similarly, when $\omega_3$ is prepared and $M_2$ measured, $e_3$ occurs with probability $\alpha$ and $e_2$ otherwise).
 }

Suppose someone is promised to receive, with probability ${\rm P}(i)$, the state {$\omega_i$} from {the set $\{\omega_0, \omega_2, \omega_3\}$} and should guess the value of $i$. Then, the probability of success when using measurement $M_1=\{e_0,\bar{e}_0\}$ is given by
\begin{equation}
     p_{\rm success}^{M_1}= {\rm P}(0) + (1/2)[{\rm P}(2)+{\rm P}(3)],
\end{equation}
since one can always guess correctly if the outcome related to $e_0$ clicks but must make a random guess between $i=2$ or $i=3$ if the other outcome clicks.
However, by using $M_2=\{e_0,\alpha e_2,\alpha e_3\}$ {the success probability is}
\begin{equation}
    p_{\rm success}^{M_2} = {\rm P}(0) + \alpha[{\rm P}(2)+{\rm P}(3)]>p_{\rm success}^{M_1},
\end{equation}
since $\alpha>1/2$. 
We see that the refined measurement $M_2$ allows for a higher probability of distinguishing between a set of states which is larger than the {size of the maximal frame}. In other words, the refined measurement $M_2$ can distinguish slightly more than $1$ bit, even though the maximal frame has size $2$ and this measurement $M_2$ coarse-grains to the distinguishing measurement $M_1$. 
If one understands coarse-graining as erasing of classical information, $p^{M_2}_{\rm success} >p^{M_1}_{\rm success}$ suggests that there was more classical information available than can be encoded onto a maximal frame. 
Such a phenomenon occurs for every unrestricted GPT built from a polygon state space with an odd number of vertices {(see also \citet{Massar_InfoPolygon}).}
This difference between the amount of classical information that can be encoded into a GPT system and the size of a maximal frame is a violation of a principle that has been called ``No Simultaneous Encoding''~\cite{Masanes_InfoUnit}.
By explicitly only allowing MCI--frames (\cref{def:MCIframe}) to characterize the classical information to be spread by an {idealized} Darwinism process (\cref{def:DarwinGPT}),
 we ensure that no such over--encoding occurs in the systems considered in this article.

\section{Pure separable objects}
\label{app:SepEffect}
\begin{lemma}
\label{LemPureProductEffects}
A pure effect is separable if and only if it is a product of pure effects.
\end{lemma}
\begin{proof}
Only one direction is non-trivial: suppose that the effect $e^{1,2,\ldots,N}$ is separable, then it can be written
\begin{align}
   e^{1,2,\ldots,N}&=\sum_i e_i^{(1)}\odot\ldots\odot e_i^{(N)}
\end{align}
where the $e_i^{(j)}$ are suitable local effects. Since $e^{1,2,\ldots,N}$ is pure, we must have $e_i^{(1)}\odot\ldots\odot e_i^{(N)}\propto e^{1,2,\ldots,N}$ for all $i$. Hence these product effects are all multiples of each other, and $e^{1,2,\ldots,N}=e^{(1)}\odot\ldots\odot e^{(N)}$ for suitable local effects $e^{(j)}$. If we could non-trivially decompose any of the $e^{(j)}$, then we could decompose $e^{1,2,\ldots,N}$, which would contradict its purity.
\end{proof}
{
For completeness, we also show that pure separable states are products of pure states.
\begin{lemma}
\label{LemPureProductStates}
A pure state is separable if and only if it is a product of pure states.
\end{lemma}
\begin{proof}
Again, only one direction is non-trivial: suppose $\omega ^{1,2,\ldots,N}$ is separable. Then, 
\begin{align}
    \omega^{1,2,\ldots,N}=\sum_i p_i \omega_i^{(1)}\odot\ldots\odot \omega_i^{(N)},
\end{align}
with $\omega_i^{(j)}$ local states and $(p_i)_i$ a probability distribution that, without loss of generality, satisfies $p_i>0$ for all $i$. Since $\omega^{1,\ldots,N}$ is pure, $(\omega^{(1)}_i\odot\ldots\odot\omega^{(N)}_i)_i$ must be all equal to each other. Thus, $\omega^{1,2,\ldots,N}=\omega^{(1)}\odot\ldots\odot\omega^{(N)}$. Again, if any $\omega^{(j)}$ were a non-pure state, we could decompose $\omega^{1,2,\ldots,N}$ non-trivially, which would be a contradiction.
\end{proof}
 }
 
 \section{Quasi-classical MCI-frames in n-gon GPTs}
 \label{App:n-gonGPTs}

{
Here we show that existence of quasi-classical MCI-frames is rare among a class of GPTs, namely, the $n$-gon theories~\cite{Janotta_PolygonSpaces}. These theories are constructed within $\reals^3$. 
The normalized state space $\Omega$ of each GPT is a regular polygon with $n$ vertices, and their effect cones satisfy the no-restriction hypothesis. 
This class of GPTs is quite rich and interesting. 
First, taking the limit $n\rightarrow \infty$, one obtains a quantum bit over the real numbers.  
Second, one can find examples of GPTs that both obey or disobey important principles: for instance, all those with odd $n$ obey strong symmetry (and therefore are strongly self-dual \cite{Mueller_Bit-Symmetry}) as opposed to those with even $n$. 
On the other hand, those with even $n$ have MCI-frames, which those  with odd $n>3$ lack (recall \cref{app:Examples} for $n=5$, or  see \citet{Massar_InfoPolygon} in general). 
As we shall see below, even among those $n$-gon with MCI-frames, only a couple have {\em quasi-classical} MCI-frames.}
 
 {
 \begin{example}[quasi-classical MCI-frames in n-gon theories]
\label{ex:ngonquasiCl}
Consider the $n$-gon GPTs~\cite{Massar_InfoPolygon, Janotta_PolygonSpaces}, whose normalized state space is a regular polygon with $n$ vertices, and whose effect cones satisfy the no-restriction hypothesis. Among the $n$-gon theories that feature MCI-frames (i.e., those with even $n$ or $n=3$), only those with $n=3, 4$ feature quasi-classical MCI-frames (see figure \ref{figure:n-gonQuasiClassical}).
\end{example}
}
\begin{figure}[bth]
\centering
\includegraphics[width=\linewidth]{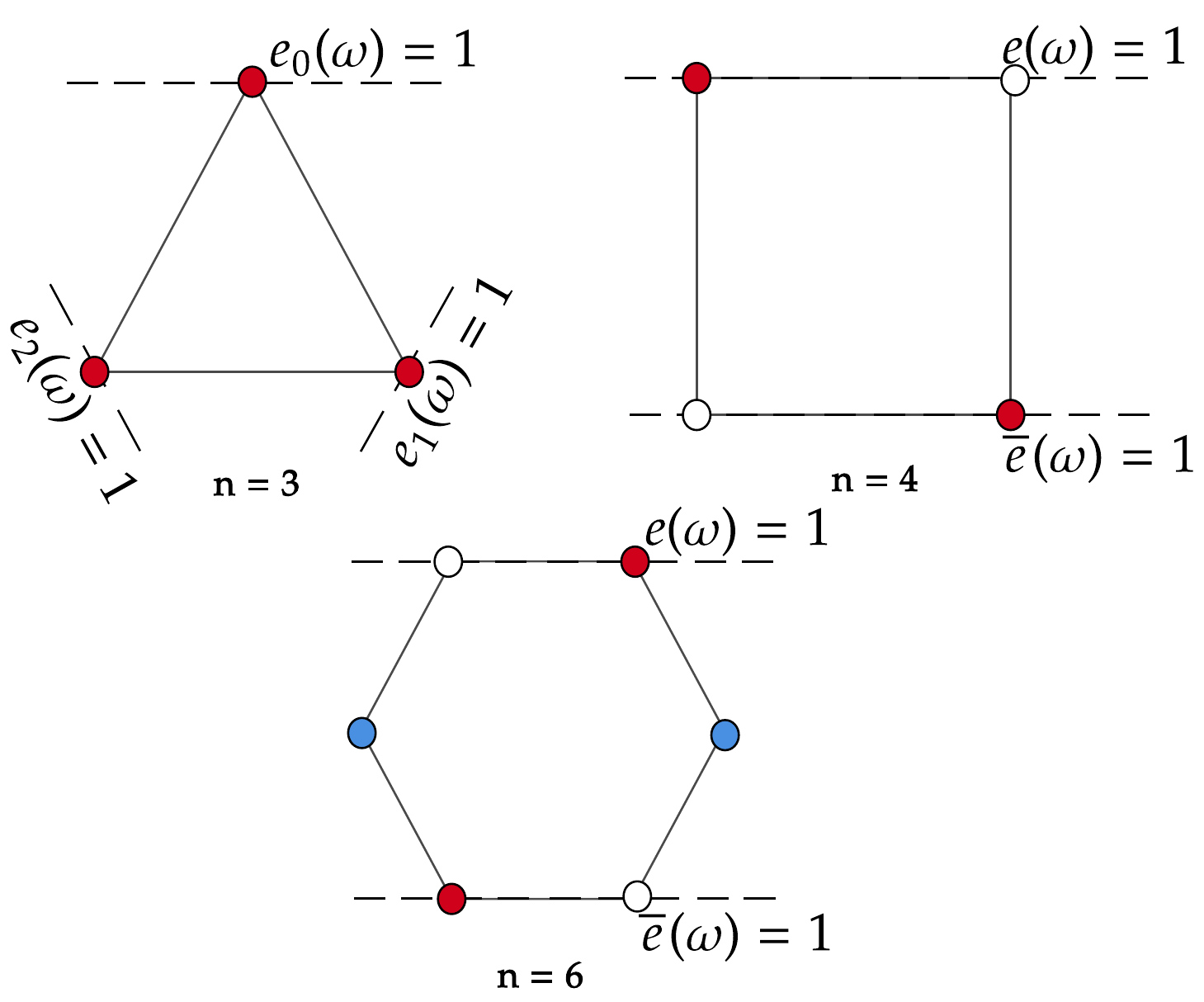}
\caption{ {\caphead{Quasi-classical MCI-frames in $n$-gon GPTs.} The red vertices are (choices of) maximal frames and the dashed lines represent the refined distinguishing measurements for that frame. White dots are pure states that assign deterministic probabilities to the distinguishing measurements and blue dots represent the remaining pure states -- which \emph{do not} assign deterministic probabilities to the MCI-frame distinguishing measurements, as they lie between the two dashed lines. For even $n\geq6$, such blue vertices appear. Therefore, the only $n$-gon GPTs that have quasi-classical MCI-frames are the classical `trit' ($n=3$) and the `gbit' ($n=4$).}
}
\label{figure:n-gonQuasiClassical}
\end{figure}
{
As we see, almost all $n$-gon theories lack quasi-classical MCI-frames. 
The unique non-classical $n$-gon theory that does carry those is the square-shaped GPT (i.e. gbits in the ``boxworld'' GPT~\cite{Barrett_GPT}), since the case $n=3$ is a classical `trit'.
This example shows that requiring quasi-classical MCI-frames -- which is a consequence of demanding idealized Darwinism without entangled states, as per Theorem~\ref{theorem:DneedsEntangledStates} -- indeed restricts the considered GPTs {quite substantially.}
}

\section{Spekkens' Toy Model}
\label{app:Spekkens}
In this appendix, we briefly review some details of Spekkens' Toy Model~\cite{Spekkens_ToyModel} (STM) and its GPT extensions~\cite{hardy_ConvexSpekkensTT,Janotta_GPTsRestriction,Garner_GPTPhase}.

\subsection{Overview}
STM is essentially a classical {local} hidden-variable model on which an {\em epistemic restriction} is imposed: no more than half the information (as measured in bits) can be known.
The simplest (and for our purposes, only) single system in this framework then consists of a so-called {\em ontic} hidden variable with four possibilities $\{1,2,3,4\}$.
Valid questions about such system can only narrow down the state to at best two possibilities (e.g.\ ``is the system in $1\lor2$ (read `1 or 2')?'') for both affirmative and negative answers to the question.
This yields three sets of mutually exclusive questions of the form ``is the system in [X]'' which we label as follows:
\begin{align}
\bra{x+} := 1 \lor 3, \qquad & \bra{x-} := 2 \lor 4, \nonumber \\
\bra{y+} := 1 \lor 4, \qquad & \bra{y-} := 2 \lor 3, \nonumber \\
\bra{z+} := 1 \lor 2, \qquad & \bra{z-} := 3 \lor 4.
\end{align}
By the rules of STM, whenever such a question is asked, the ontic state must be randomized within the supporting set of states consistent with the answer to the question.
For example, an affirmative answer to question $\bra{z+}$ will randomize the ontic state of the system to $1$ or $2$.
This randomization ensures we cannot find the exact ontic state, say, by asking two different questions in a row -- while maintaining the property that if we ask the same  question twice in a row, we will get the same answer.
Thus, one may define a set of maximum--knowledge {\em epistemic states} in one-to-one correspondence with the affirmative answer to these questions, labeled, e.g.,\ as $\ket{x+} = 1\lor 2$.
(STM also admits a ``unit'' question $u:=$ ``is the system in $1\lor2\lor3\lor4$?'' to which the answer is always affirmative;
 similarly, there is also a maximally mixed state, in which the ontic state can take any value with the same probability.)
 
The ontic state of a composite system is formed by taking the Cartesian product of each constituent system's ontic state (written for $a$ and $b$ as $ab$).
The allowed epistemic states in this context then are those that satisfy the epistemic restriction both on the entire system, and also any subsystem thereof.
Thus, a two-system epistemic state must admit at least four ontic possibilities.
In addition to the Cartesian product of single system states, this also allows for ``entangled'' states, such as $11\lor 22 \lor 33 \lor 44$, where even though the local marginal states are maximally mixed, perfect correlation is guaranteed if the same measurement is made on both systems.
On the other hand, a state such as $11\lor 12 \lor 33 \lor 44$ is forbidden.
This is because should the $\bra{z+}$ measurement on the second system be answered in the affirmative, then the first system is definitely in state $1$, which violates the epistemic restriction.
It can thus be seen that STM is {\em self--dual by construction}: every maximum--knowledge measurement outcome can be uniquely identified with a maximum--knowledge epistemic state~\cite{Pusey_Stabilizer}.

Transformations in the theory are performed by permuting the underlying hidden variable, in such a way that no valid epistemic state is taken to an invalid state.
For single systems, every permutation is valid -- but this is not the case for multipartite systems.
Since these permutations are a finite group, when searching for a transformation that achieves a desired outcome (e.g.\ exhibits Darwinism), one can (with computer assistance) exhaustively search through possible transformations to find one that achieves the desired aims -- or otherwise rule out its existence entirely~\cite{Garner_DPhil}.
However, by formalizing the similarity between STM and the stabilizer subset of quantum mechanics, \citet{Pusey_Stabilizer} enables an elegant sufficient condition for the existence of a transformation, which we will subsequently describe.

\subsection{GPT Extension}
\label{app:STM_GPT}

First, however, let us remark on the extension of STM into the GPT framework.
In particular, STM defines a discrete state space with a finite number of states -- so in order to treat it as a GPT, we must make it continuous.
This is done in the obvious way:
 we treat the questions such as ``is the system in $1\lor2$?'' as an effect, and then admit all convex combinations of such effects. 
A complete (i.e.\ at least one question answers in the affirmative for any state) and mutually exclusive (i.e.\ no more than one question answers in the affirmative) set of questions maps to a set of effects that form a normalized measurement (i.e.\ will sum to the unit effect).
Meanwhile, each set of epistemic states of maximum knowledge with no overlap in their ontic variable support (e.g. $\{1\lor2, 3\lor4\}$) form maximal frames, in which the maximum-knowledge epistemic states are extremal. 
We then allow convex combinations of such states as ``mixed'' states, yielding a theory dubbed STM--GPT. The set of allowed transformations on the theory are then defined as exactly those allowed on the (non-GPT) STM, and due to linearity, each of these uniquely extends into a transformation on the STM--GPT state space\footnote{This implies that not all symmetries of the state space of STM-GPT belong to the group of allowed transformations, {$\mathcal{T}$}. 
For instance, the rotation in the $z$-axis which permutes $\ket{y+}\mapsto\ket{x+}\mapsto\ket{y-}\mapsto\ket{x-}\mapsto\ket{y+}$ is a symmetry of the octahedron but is not an allowed transformation in the ontic state space (see figure \ref{fig:SpekkensBit}).}.

\begin{figure}[bth]
\centering
\includegraphics[width=0.65\linewidth]{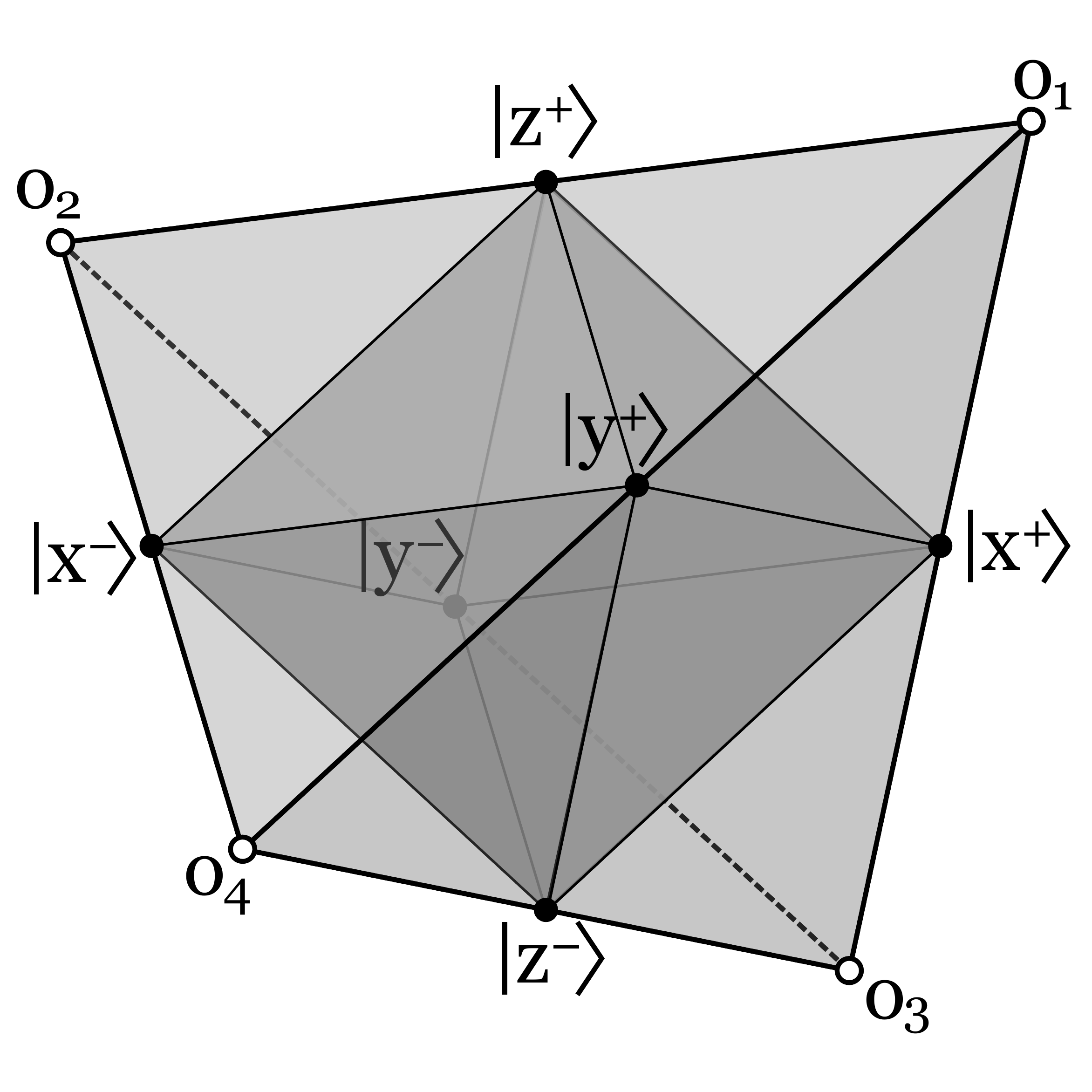}
\caption{ \caphead{Normalized states of a Spekkens' bit.}
The tetrahedron is the normalized slice of $\reals^4$ corresponding to the underlying classical ontic variable, with basis states $\{\vec{o}_1,\vec{o}_2,\vec{o}_3,\vec{o}_4\}$.
The pure epistemic states correspond to the half-way points between these ontic states.
The valid epistemic states of the theory are these states' octahedral convex hull.
\label{fig:SpekkensBit}
}
\end{figure}

One representation of a single system in STM--GPT in $\reals^4$ is to identify each {\em ontic} state with a Cartesian vector, $\vec{o}_1 := (1,0,0,0)\trans$, $\vec{o}_2 := (0,1,0,0)\trans$, $\vec{o}_3 := (0,0,1,0)\trans$, $\vec{o}_4 := (0,0,0,1)\trans$,  and then write each epistemic state $x\lor y$ as the vector $\frac{1}{2} \left( \vec{e}_x + \vec{e}_y \right)$ (see \cref{fig:SpekkensBit}).
Here $A=\reals^4$ and $\Omega_A$ is the convex combination of such (geometrically: this is the octahedron formed by connecting the midpoint of every line in a tetrahedron~\cite{Garner_GPTPhase}).
As observed in \citet{Janotta_GPTsRestriction}, the unrestricted dual of this space is cubic (i.e.\ a gbit) -- but STM does not follow the no-restriction hypothesis.
Rather, instead the space of effects can be represented by {\em exactly the same} vector space (carrying forward the self-duality-by-construction of STM),
 where the self-dualizing inner product $\langle e, \rho \rangle := 2 e \cdot \rho$ is directly proportional to the Euclidean inner product on the real vector spaces.

An analogous representation can also be formed for $n$ STM--GPT systems in $\reals^{4^n}$. 
Take the Cartesian product $\{\vec{o}_1, \vec{o}_2,\vec{o}_3,\vec{o}_4\}^{\otimes n}$ to find the set of ontic states,
 and likewise define the epistemic set as valid (as per above) mixtures thereof.
For example, $11\lor22\lor33\lor44$ is represented here as $\frac{1}{4}\left(\vec{o}_{11} + \vec{o}_{22} + \vec{o}_{33} + \vec{o}_{44}\right)$.
Meanwhile, product states of lower-dimensional STM--GPT systems are simply found by the tensor product.
For example, $1\lor2 \otimes 1 \lor 3 \equiv 11 \lor 13 \lor 21 \lor 23$ satisfies 
$\frac{1}{2}\left(\vec{o}_1 + \vec{o}_2\right) \otimes \frac{1}{2}\left(\vec{o}_1 + \vec{o}_2\right) = \frac{1}{4}\left(\vec{o}_{11} + \vec{o}_{13} + \vec{o}_{21} + \vec{o}_{23}\right)$.
This also allows for a self-dualizing inner product: $\langle \vec{e}, \vec{\rho} \rangle := {2^n} \vec{e} \cdot \vec{\rho}$.

\subsection{Stabilizer Formalism}
\label{app:SpekkensStab}
Stabilizer groups originate in group theory, but have been adapted for use in quantum theory in the context of error-correcting codes and measurement-based quantum computation,
 as they provide concise ways to describe certain high-dimensional {quantum} states.
Essentially, a transformation $T$ is said to stabilize a state $\ket{\psi}$ if $T\ket{\psi} = \ket{\psi}$~\cite{Nielsen_Chuang}.
Listing enough simultaneous stabilizing transformations may be enough to uniquely define a state (up to global phase):
 for example, the only two qubit state stabilized by both $\sigma_x\otimes \sigma_x$ and $\sigma_z\otimes\sigma_z$ is the Bell state $\ket{\Psi}=\frac{1}{\sqrt{2}}\left(\ket{00}+\ket{11}\right)$.
The {\em stabilizer subset of quantum theory} are exactly the $n$ qubit states that can be so described, when the stabilizers are taken from the Pauli group $P_n := \{\pm\id,\pm\sigma_x,\pm i\sigma_x, \pm\sigma_y,\pm i\sigma_y,\pm\sigma_z,\pm i\sigma_z\}^{\otimes n}$.

STM(--GPT) shares many similarities with (the convex hull of) quantum stabilizer states~\cite{Pusey_Stabilizer}. 
For instance, a qubit has six distinct pure qubit stabilizer states (stabilized by the Hermitian elements $\pm \sigma_x$, $\pm \sigma_y$, and $\pm \sigma_z$).
Meanwhile, for an STM bit (using the GPT representation above), we can similarly define three ``observable'' matrices:
\begin{align}
X & := \diag \left(1,-1,1,-1\right), \nonumber \\ 
Y & := \diag \left(1,-1, -1, 1\right), \nonumber \\
Z & := \diag \left(1,1,-1,-1\right),
\end{align}
such that for each measurement, there is a unique (pure) epistemic state corresponding to the $1$ and $-1$ eigenvector from each (e.g. $X \ket{x+} = \ket{x+}$) -- and this covers all pure epistemic states.
We can identify each of $X$, $Y$ and $Z$ respectively with the ontic state permutations
\begin{align}
X \leftrightarrow  3412, \quad  Y \leftrightarrow  4321, \quad Z \leftrightarrow 2143,
\end{align}
along with an identity element $I := \diag(1,1,1,1) \leftrightarrow 1234$.
Then $\{I,X,Y,Z\}$ together with matrix multiplication is the Klein four-group $V$ and is isomorphic to the permutation subgroup $\{1234,3412,4321,2143\}$.
The Cartesian product of these matrices with $\mathbb{Z}_2 = \{+1,-1\}$ forms the {\em toy stabilizer group} $G := \mathbb{Z}_2 \otimes V = \{\pm I, \pm X, \pm Y, \pm Z\}$.

Unlike the Pauli group, this group is Abelian with $XZ=ZX=Y$ (cf.\ $\sigma_x \sigma_z = - \sigma_z\sigma_x = -i\sigma_y$).
For $n$ bit systems, we denote the application of $T\in V $ to the $k^{\rm th}$ system as $T_k := I^{\otimes(k-1)} \otimes T\otimes I^{\otimes (n-k)}$.
Finally, let us define the map $m: V^n \to P^n$ that makes an obvious identification between {STM stabilizers and quantum stabilizers} (e.g.\ $m: I_1 X_2 \mapsto \id\otimes\sigma_x$).
 
Now we may use the result of \citet{Pusey_Stabilizer}: if a set of independent quantum stabilizers $m(R^1), m(R^2), \ldots m(R^k)$ describes a unique quantum state, then $R^1, R^2, \ldots, R^k$ describes a unique epistemic state in STM.
Moreover, if a map on a set of quantum stabilizers $T: m(A^1) \mapsto m(B^1), \ldots, m(A^k) \mapsto m(B^k)$ defines a unitary quantum transformation and $m(A^1)\ldots m(A^k)$ are a {\em canonical generating set},
 then $A^1 \mapsto B^1, \ldots, A^k \mapsto B^k$ defines a valid STM transformation.
The full definition of canonical generating set is complicated, but for our purposes, it suffices to note that $\{X_1, \ldots X_k, Z_1, \ldots Z_k\}$ is one such set.
With the aid of these sets, we can construct the FAN transformation (defining how it acts on each $\mathcal{X}_k/\mathcal{Z}_k$, that broadcasts information about the measurement $\{\bra{z+},\bra{z-}\}$ to the environment, (see equation~\eqref{eq:stabCNOTmany}).
\vspace{0.5cm}

\subsection{{STM is not strongly symmetric, nor does it have a decoherence map}}
\label{app:SpekkSuff}

In this section, we show that stabilizer quantum theory and (GPT-)STM fail to admit a decoherence map (in the sense of \citet{Richens_DecoherenceGPTs}, as adapted in \cref{def:DecoherenceMaps}), and similarly neither theory obeys strong symmetry.
\begin{lemma}
\label{lemma:StabilizerXDecoherence}
Stabilizer quantum states do not admit a decoherence map.
\end{lemma}
\begin{proof}
 By counterexample. 
Consider the classical $3$-bit control-control-NOT gate that flips the third bit only if the first two bits are in state 1, and otherwise does nothing. 
This corresponds to a Toffoli gate in the quantum circuit, which is {\em not} a member of the Clifford group~\cite{Gottesman_stab}, and hence not a valid quantum stabilizer transformation. 
This violates condition $3$ of \cref{def:DecoherenceMaps}: there is a classical reversible transformation that cannot be induced by a transformation in the theory.
\end{proof}

Analogously, there is a classical transformation that {cannot be implemented in STM as well}:
\begin{lemma}
\label{lemma:STMXDecoherence}
Spekkens' Toy Model does not admit a decoherence map.
\end{lemma}
\begin{proof}
By counterexample. Consider the classical 3-bit transformation where bits 2 and 3 are swapped if bit 1 is set. 
This also forms a valid classical transformation in these bits. 
It is shown exhaustively by \citet{Garner_DPhil} that STM does not have a 3-bit controlled SWAP. 
This then amounts to a valid classical transformation that cannot be induced within STM, violating condition $3$ of \cref{def:DecoherenceMaps}.
\end{proof}
As we have argued in the main text, the ability to induce any reversible classical transformation on a frame's states (effects) is a necessary condition for strong symmetry on states (on effects) in a GPT system. Thus, we conclude:
\begin{lemma}
Neither stabilizer quantum theory nor Spekkens' toy model satisfy strong symmetry (on states or effects).
\end{lemma}
\end{document}